\documentclass[titlepage,twoside,11pt,letter]{article}
\usepackage[margin=1.2in]{geometry}
\usepackage{enumerate}

\usepackage{times}
\usepackage{amsmath}
\usepackage{amsfonts}
\usepackage{amssymb}
\usepackage{graphicx}
\usepackage{mathrsfs}
\usepackage{tabularx}
\usepackage{caption}
\usepackage{bbold}
\usepackage{color}
\usepackage{fancybox}
\usepackage{verbatim}
\usepackage{hyperref}
\usepackage{tikz}
\usetikzlibrary{arrows,automata}
\usetikzlibrary{trees}

\def\bR{\mathbb{R}}

\def\b1{\mathbb{1}}

\def\eE{\mathsf{E}}

\def\eP{\mathsf{P}}

\def\fr{{\mathfrak{r}}}

\def\fN{{\mathfrak{N}}}

\def\cA{\mathcal{A}}
\def\cB{\mathcal{B}}
\def\cC{\mathcal{C}}

\def\cE{{\mathcal{E}}}

\def\cL{\mathcal{L}}

\def\cV{{\mathcal{V}}}
\def\cY{{\mathcal{Y}}}
\def\cZ{{\mathcal{Z}}}

\def\sF{{\mathscr{F}}}
\def\sG{{\mathscr{G}}}

\def\sJ{{\mathscr{J}}}

\def\e1{\mathsf{1}}

\def\ofx{(x)}

\def\of0{(0)}

\def\d{\partial}
\def\tr{\mathrm{tr}}

\def\bf0{\mathbf{0}}
\def\cp1{\mathbb{CP}^1}
\def\matn{\mathrm{Mat}_n(\bR)}

\newtheorem{theorem}{Theorem}[section]

\newtheorem{proposition}[theorem]{Proposition}

\newenvironment{proof}[1][Proof]{\noindent\textbf{#1:} }{\hfill \rule{0.5em}{0.5em}}

\DeclareMathOperator{\diag}{diag}

\begin{document}
	
\title{\textbf{Managing counterparty credit risk via BSDEs}}
\author{\textbf {Andrew Lesniewski and Anja Richter\footnote{andrew.lesniewski@baruch.cuny.edu, anja.richter@baruch.cuny.edu} }\\
Department of Mathematics\\
Baruch College\\
One Bernard Baruch Way\\
New York, NY 10010\\
USA}
\date{\today}

\maketitle

\begin{abstract}
We discuss a general dynamic replication approach to counterparty credit risk modeling. This leads to a fundamental jump-process backward stochastic differential equation (BSDE) for the credit risk adjusted portfolio value. We then reduce the fundamental BSDE to a continuous BSDE. Depending on the close out value convention, the reduced fundamental BSDE's solution can be represented explicitly or through an accurate approximate expression. Furthermore, we discuss practical aspects of the approach, important for industry applications: (i) efficient numerical methodology for solving a BSDE driven by a moderate number of Brownian motions, and (ii) factor reduction methodology that allows one to approximately replace a portfolio driven by a large number of risk factors with a portfolio driven by a moderate number of risk factors.
\end{abstract}

\tableofcontents

\newpage

\section{\label{introSec}Introduction}

Counterparty credit risk has recently moved to the forefront of attention of the financial industry and government regulatory agencies. Prior to 2007, financial institutions used to make provisions for possible counterparty defaults based on \emph{ad hoc} assumptions and proprietary metrics. As a result of the financial crisis of 2007 - 2008 and the subsequent stricter regulatory requirements, careful analysis and more rigorous modeling methodologies of counterparty credit risk have become a centerpiece of prudent risk management practice. In particular, these developments led to the establishment of special business units within financial institutions whose sole purpose is to monitor and mitigate counterparty credit risk. 
Summaries of various approaches to counterparty credit risk modeling can be found in \cite{cbb14}, \cite{g15}, \cite{greg15}.

In a counterparty credit risk model, one considers two parties, namely a dealer (``a bank'') and a client (``a counterparty'') such as an asset manager, hedge fund or corporate, which have a portfolio of over the counter (OTC) transactions between them\footnote{Note that there is a tendency towards central clearing of transactions, however a large portion of transactions is still done OTC.}. Over the time horizon of the transactions, a default of either the bank or the counterparty may occur. The objective of the model is to put a dollar value on this default risk and develop strategies to mitigate it. The portfolio of transactions is typically structured as a collection of netting sets. From the risk management perspective, the cash flows within the same netting set are allowed to partially offset each other.

Generally, counterparty credit risk is managed through maintaining appropriate reserves of cash and high quality financial assets. In the case of Credit Support Annex (CSA) \cite{Isda16} transactions, the two parties exchange and manage collateral. Additionally, a fully fledged counterparty credit risk model takes into account the impact of margining and funding costs. 

In this paper, we are concerned with a general approach to counterparty credit risk modeling. A natural modeling framework for counterparty credit risk, which we adopt in our discussion, is provided by backward stochastic differential equations (BSDEs). In particular, this approach allows for an efficient and fully consistent description of counterparty default valuation along with its hedging strategy, and can be easily implemented numerically. BSDEs, in the context of finance, have been studied in the literature since the 1990's, see \cite{kpq97}, and have recently been applied to counterparty credit risk modeling, see e.g \cite{cbb14}, \cite{bcs16}. 

Our approach extends the partial differential equations (PDE) framework developed earlier by Burgard and Kjaer \cite{bk12}, see also \cite{gk14}, \cite{gkd14}, \cite{kgb15}. Their methodology is based on a dynamic replication strategy involving the underlying assets as well as both parties' credit sensitive zero coupon bonds, and it culminates in the fundamental PDE for the counterparty credit sensitive value of an instrument. They also provide closed form solutions to the fundamental PDE in a number of relevant cases. Here, we extend their approach to a more general multivariate diffusion setting including stochastic interest rates and default intensities. 

Following a similar strategy, we formulate the counterparty credit risk model in purely probabilistic terms. The dynamic replication approach leads to a jump diffusion BSDE with a random time horizon. The reason why a jump may occur in the BSDE is that the possible default of either the counterparty or the bank itself (and associated changes in the value of the portfolio) can occur at a random time. Equations of this type are usually very challenging to solve numerically  \cite{gl15}. Motivated by the method developed in \cite{kln13}, we reduce our fundamental jump diffusion BSDE to a continuous BSDE with a fixed terminal time. 

In case of a counterparty or bank default, their assets are liquidated and the claims of their creditors are settled. Different close out conventions are used to determine the value of the portfolio of transactions between the two parties at default. Here, we discuss in detail two explicit close out conventions. The solutions to the corresponding fundamental BSDEs are interpreted in terms of the commonly used valuation adjustments (the ``XVAs''). In one of the cases, a closed form representation of the solution is available, while in the other one, we derive an accurate approximation to the solution.

An important practical aspect is the computational feasibility of the solution to the fundamental BSDE. In practice, the portfolio of transactions between the dealer and the counterparty may consist of hundreds (or thousands) of different positions. Naive modeling of such a portfolio leads to a high dimensional stochastic system. We discuss an approximation method that allows us to reduce the full model to model driven by a moderate number of risk factors. We then prove that reduced factor BSDE indeed approximates the fundamental BSDE. Moreover, we briefly discuss the question of choosing an appropriate pricing measure.

For continuous diffusion BSDEs, efficient numerical schemes have been developed, see e.g. \cite{bt04}. We apply a scheme to numerically solve the reduced fundamental BSDE which requires an efficient methodology to calculate conditional expected values. To this end, we use a variant of the Longstaff-Schwartz \cite{ls01} regression technique that is adapted to our setting.  Specifically, we propose the use of Hermite polynomials as basis functions. Hermite polynomials have various advantages, for example they exhibit an addition formula and a martingale property, which provide practical means for calculating conditional expectations. Additionally, using the Hermite architecture allows us to represent a part of the solution of the BSDE in terms of explicit expressions. We continue the study of the numerical aspects of the problem in \cite{llr16}.

The paper is organized as follows. We present our model setup and develop the fundamental equations describing the value of a portfolio with and without counterparty credit risk in Section \ref{bsdeSec}. The special case of deterministic interest rates and dividends, leading to a Burgard-Kjaer type of PDE, is covered in Section \ref{sec:BKPDE}. In Section \ref{sec:valueAdj}, we discuss the total value adjustments for the two most common close out values. We then turn our attention to practical issues such as the choice of pricing measure and risk factor reduction in Section \ref{sec:PricingMeasureFactorReduction}. Finally, we propose a numerical algorithm to solving the fundamental BSDE using Monte Carlo methods in Section \ref{numSec}. The key technical results are proved in the Appendix.

\section{\label{bsdeSec}The fundamental equations}

\subsection{The model}

In this Section we present a mathematical formulation of the problem of modeling counterparty credit risk between two counterparties: $\cB$ (a ``bank'') and $\cC$ (a ``client''). The equations written below assume the perspective of the bank $\cB$. We study contracts on an asset $S$ between the bank and its counterparty $\cC$, both of which may default. In our setting neither of the two possible defaults have an effect on the asset $S$.

With this situation in mind, we consider a probability space $(\Omega, \sG, \eP)$ and let the market filtration $\sF := (\sF_t)_{t \geq 0}$ be generated by an $n$-dimensional Brownian motion $W$ and augmented by all $(\sG, \eP)$-null sets. The information on default events is represented by the filtration $\sJ:= (\sJ_t)_{t \geq 0}$ to be specified later. We define the enlarged filtration $\sG := (\sG_t)_{t \geq 0}$ by $\sG_t := \sF_t \vee \sJ_t $, for all $t \geq 0$. Then $(\Omega, \sG,\eP)$ is a filtered probability space satisfying the usual conditions. See for example \cite[VI.3]{p04} for the properties of progressive enlargements of filtrations.

Let $S=(S_1,\ldots,S_n)^\top$ denote the price process, namely an $n$-dimensional Markov process with infinitesimal generator $\mathcal A_t$. The dynamics of the asset under the measure $\eP$ are
\begin{equation}\label{assetDyn}
\begin{split}
dS_t &= \mu(t,S_t) dt + \sigma(t,S_t) dW_t,\\
S_0 &= s_0.
\end{split}
\end{equation}
Here, $\mu: \bR_+ \times \bR^n \to \bR^n$ and $\sigma : \bR_+ \times \bR^n \to \matn$ are deterministic functions and $W$ is a standard $n$-dimensional Brownian motion. We assume the usual conditions for $\mu$ and $\sigma$ to guarantee the existence and uniqueness of the strong solution to this SDE. Additionally, $s_0$ is an $n$-dimensional real vector.

We consider a netting set $\mathfrak{N}$ consisting of $n$ instruments which may comprise of derivatives and cash type with underlying asset $S$. From the counterparty risk management perspective, transactions with a fixed counterparty are allowed to be offset against each other within a netting set. The composition of netting sets is set forth by bilateral agreements. We denote the time $t$-value of the netting set by $\mathfrak{N}(S_t)$.

Additionally, we consider the riskless (default-risk-free) bank account $B^R$. Its dynamics is given by
\begin{equation} 
\begin{split}
dB^R_t &= r_t B^R_t dt, \\
B^R_0&=1,
\end{split}
\end{equation}
where $r$ is a stochastic credit riskless interest rate. In the following, we assume that $r$ as well as all other rates and dividends are adapted with respect to the market filtration $\sF$. From the financial perspective, $r$ is the rate paid by a theoretical ``central margin account'', and can be proxied by the OIS rate.

In our model we allow for both, the bank and its counterparty to default. We let $\tau^\cB$ and $\tau^\cC$ be the random default times of the bank and the counterparty respectively and denote their indicator processes by counting processes
\begin{equation*}
J^\cB_t = \e1_{\tau^\cB \leq t},
\end{equation*}
and
\begin{equation*}
J^\cC_t = \e1_{\tau^\cC \leq t}.
\end{equation*}
The natural filtration generated by $J^\cB$ and $J^\cC$ constitutes the default event filtration $\sJ$, i.e. $\sJ_t := \sigma ( J^\cB_s, J^\cC_s : s \leq t )$. Both counting processes $J^\cB$, $J^\cC$ are assumed to be Cox processes, i.e. they have stochastic, time-dependent intensities $\lambda^\cB$, $\lambda^\cC$, i.e.
\begin{equation*}
\lambda^\cB_t dt = \eE[dJ^\cB_t\,|\,\mathcal G_{t-}],
\end{equation*}
and
\begin{equation*}
\lambda^\cC_t dt = \eE[dJ^\cC_t\,|\,\mathcal G_{t-}].
\end{equation*}

We denote the bank's and counterparty's default risky, zero-recovery ZCBs $P^\cB$ and $P^\cC$ with respective maturities $T^\cB$ and $T^\cC$. They follow the dynamics
\begin{equation*}
\begin{split}
dP^\cB_t &= r^\cB_t P^\cB_t dt - P^\cB_{t-} dJ^\cB_t, \\
dP^\cC_t &= r^\cC_t P^\cC_t dt - P^\cC_{t-} dJ^\cC_t.
\end{split}
\end{equation*}
Here, the adapted stochastic processes $r^\cB, r^\cC$ are the yields on $P^\cB$ and $P^\cC$, respectively.

The legal framework for OTC derivative trades is provided by the {\it ISDA Master Agreement} \cite{Isda16} which contains a common core and adjustable terms to be agreed on by both parties. The agreement aims at mitigating (counterparty) risk by documenting aspects like netting, collateral cash-flows, default events and the close out process (see e.g. \cite{g15,greg15}). In a world, where usually several transactions with a counterparty occur, netting allows the two parties to offset what they owe to one another. 

The financial crisis in 2008 fueled the need for derivative pricing methodologies that include aspects of counterparty credit risk. The goal is to find the value $\hat V$ of a netted portfolio of derivatives $\mathfrak{N}$ on $S$ allowing for both the bank $\cB$ as well as its counterparty $\cC$ to default. More precisely, we let $\hat V=\hat V(t,T)$ denote the time $t$ value of the netting set $\mathfrak{N}$ with time horizon $T$, which is assumed to satisfy $T \leq \min(T^\cB,T^\cC)$. The value of the netting set $\mathfrak{N}$ between the bank and the counterparty without counterparty default risk is denoted by $V=V(t,T)$. The difference between these two values, denoted by $A$, is called the total valuation adjustment, i.e. $\hat V = V+A$.

To mitigate counterparty credit risk, the two parties exchange collateral in form of cash or high quality financial instruments. The mechanics of collateral in OTC transactions are specified by the CSA which is usually amended to the ISDA Master Agreement. Typically, one distinguishes between two kinds of collateral, namely the initial margin and the variation margin. The former is posted by both counterparties without any netting taking place\footnote{An important exception to this rule is when the counterparty is a central counterparty (CCP). In this case, only clearing members post collateral with the clearing house.}. The amount of initial margin is calculated using risk based methods such as VaR, CoVaR or stress tests. These calculations involve historical simulations, typically using an observation window of 1 to 5 years, see e.g. \cite{g15}. To ensure that the collateral can be retrieved in a default scenario, it is segregated and cannot be used to, for example, fund other positions. In contrast, variation margins, which we denote as $X$, are calculated frequently based on the market value of the transactions, and they can be netted and rehypothecated. We denote the initial margin posted by the bank $\cB$ to the counterparty by $I^{TC}$, whereas the amount from the counterparty is denoted by $I^{FC}$. The margin value adjustment (MVA) is the bank's cost of posting the initial margin $I^{TC}$ to the counterparty over the time $[0,T]$. The guidelines for margin requirements are set forth in \cite{b15}. 

While the collateral ensures that the defaulting party partially meets its contractual obligations in the event of a default, regulatory capital is designed to help the surviving party manage a potential loss arising from the default.
Specific guidelines for calculating the regulatory capital $K$, the amount of reserves the bank must hold, are provided in the {\it Basel III} document which was introduced in 2010. The cost of holding this capital over the time $[0,T]$ is called capital value adjustment (KVA).

Another component of the total value adjustment is the funding value adjustment (FVA). Essentially, the FVA accounts for the costs of funding of uncollateralized (or partly collateralized) positions. There is a debate around how the funding adjustment should be treated in a framework like ours, see \cite{greg15}, pp. 349-356, and \cite{hw14}, \cite{ads16}.

In the following we first find the default free value $V$ of the contract in terms of a conditional expectation. We present the cash-flow netting at default and the cash flows associated with the portfolio that is set up to replicate the value $\hat V$. Then the fundamental BSDE describing the counterparty credit risky value $\hat V$ is derived.
Note that, from now on, all interest rates are considered adapted stochastic processes unless stated otherwise.

\subsection{Counterparty credit risk free value}
\label{sec:valueV}

The value $V$ of the netting set $\mathfrak N$ without credit counterparty risk depends solely on the price process $S$ and is derived similarly to the price in the classical Black-Scholes model. As usual, we can set up a self-financing portfolio $\Pi$ replicating $V$. However, contrary to the classical pricing theory, the position in $S$ is not funded at a single riskless rate $r$, but instead financed by a repurchasing agreement (repo) at a repo rate $q^S$. The securities are pledged as collateral against cash to purchase the securities. In practice, a haircut is applied to the amount of cash received against the collateral, that is the loan size is smaller than the current face value of the pledged securities. The amount of haircut depends on the quality of the securities. Here, for simplicity, we shall assume a zero haircut; it is straightforward to adapt our calculations to accommodate for non-zero haircuts. 

The replicating portfolio $\Pi$ consists of a position in $\delta$ units of $S$, $\varphi^S$ units in a repo cash account $B^S$, and $\varphi^R$ units of $B^R$, i.e. we must have 
\begin{equation} \label{repPortfolioV}
\begin{split}
V_t &= \Pi_t \\
& = \delta^\top_t S_t  + (\varphi^S_t)^\top B_t^S + \varphi^R_t B_t^R ,
\end{split}
\end{equation}
for all times $t \in [0,T]$. Note that $\delta$ and $\varphi^S$ are $n$-dimensional vectors. We describe the different cash flows in the following:

\noindent 
\emph{Securities funding:} We assume the position in the security $S$ is exclusively financed by the (default-risk-free) repo cash account $B^S$, which means, we always have
\begin{equation} \label{securitiesFunding}
 \delta^\top_t S_t = - 	(\varphi^S_t)^\top B^S_t.
\end{equation}
This equation stems from the fact that if we enter into a long position in $S$, i.e. $\delta >0$, we need to finance this buy by receiving a fully collateralized loan through a repo agreement, i.e. $\varphi^S<0$. On the other hand, if we sell $S$, i.e. $\delta <0$, we invest the received cash into $B^S$, i.e. $\varphi^S>0$.
The cash account $B^S$ accrues at the repo rate $q^S$ and decreases at the dividend yield $\gamma^S$ of the underlying asset $S$, i.e. it evolves according to
\begin{equation} \label{dB^S}
\begin{split}
dB^S_t &= \diag(q^S_t - \gamma^S_t) B^S_t dt, \\
B^S_0&=1,
\end{split}
\end{equation}
where $\diag(a)$ is the diagonal matrix whose diagonal is given by the vector $a$. Recall that, according to our assumption, all rates, like $\gamma^S$ and $q^S$ here, are adapted stochastic processes.

\noindent 
\emph{Riskless deposit:} From equations \eqref{repPortfolioV} and \eqref{securitiesFunding}, we see that the amount $V$ is financed/earns the riskless rate $r$, more precisely we get
\begin{equation} \label{risklessDeposit}
V_t = \varphi^R_t B_t^R.
\end{equation}

The self-financing condition then implies that the replicating portfolio $\Pi$ has the dynamics
\begin{equation*}
\begin{split} 
dV_t &=d\Pi_t \\
& = \delta_t^\top dS_t + (\varphi^S_t)^\top dB_t^S + \varphi^R_t dB_t^R,
\end{split}
\end{equation*}
which, together with the dynamics for $S$, $B^S$ and $B^R$, leads to
\begin{equation*}
\begin{split} 
dV_t 
&= \delta_t^\top  \sigma(t,S_t) dW_t + \big(\delta_t^\top \mu(t,S_t) - (\varphi^S_t)^\top \diag(\gamma^S_t -q_t^S) B^S_t + \varphi^R_t r_t B^R_t  \big) dt, \\
&= \delta_t^\top  \sigma(t,S_t) dW_t + \big(\delta_t^\top \mu(t,S_t) + \delta_t^\top \diag(\gamma^S_t -q_t^S) S_t + r_t V_t \big) dt,
\end{split}
\end{equation*}
where we have used equations \eqref{securitiesFunding} and \eqref{risklessDeposit} in the last equality.

If we now set $Z_t^\top = \delta_t^\top \sigma(t,S_t)$, we can formulate the default free portfolio dynamics in terms of a BSDE with terminal value $\mathfrak N(S_T)$, i.e. 
\begin{equation} \label{BSBSDE}
\begin{split} 
-dV_t &=  \big(- Z_t^\top\sigma(t ,S_t)^{-1} (\mu(t ,S_t) +  \diag(\gamma^S_t-q^S_t) S_t) - r_t V_t \big)  dt - Z_t^\top dW_t,
\\ 
V_T &= \mathfrak N(S_T) .
\end{split}
\end{equation}
The value $V$, namely the first part of the solution $(V,Z)$ of this equation, can be found explicitly, as the driver is linear in $V$ and $Z$. More precisely, from Appendix \ref{app:LinBSDE} we obtain 
\begin{equation} \label{solV}
V_t = E_t \Big[ e^{- \int_t^T r_u du} \, \Gamma_{t,T} \, \mathfrak{N}(S_T)  \Big],
\end{equation}
where
\begin{equation}\label{gammaDef}
\Gamma_{t,T} = \cE\Big(-\int_t^T\big(\sigma(u ,S_u)^{-1} (\mu(u ,S_u) +  \diag(\gamma^S_u-q^S_u) S_u) \big)^\top dW_u \Big).
\end{equation}
The equation above extends the classic Black-Scholes model in that it explicitly accounts for position financing cost.
In the following we will regard the process $V$ as a known input into the counterparty credit risk model.

\subsection{Close out netting}
\label{sec:CloseOut}

Over the lifetime of the portfolio of transactions, the bank or the counterparty may default. At the time of default, the counterparty credit risk adjusted value of the portfolio $\hat V$ is determined by the terms specified in the ISDA master agreement and can take several forms. In general the value at default is impacted by the party that defaults first, the close out value $M$ and the collateral $I$ and $X$. In the literature, see \cite[section 3.1.1]{g15} for more background information, several different conventions to determine the value at default can be found.

In the following we use the notation $x^+=\max\{ x,0\}$, and $x^- = \min \{ x,0\}$\footnote{Note that we are following here the convention used in the financial literature, as opposed to the convention used in the mathematical literature according to which $x^- = \max \{ -x,0\}$.}. In a situation where the counterparty defaults, the bank is already in possession of the collateral $X+I^{FC}$. Now, if the unsecured value $M - (X+I^{FC})$ is negative, i.e. the bank owes money to the counterparty, the bank has to pay the full outstanding amount $(M - X -I^{FC})^-$. Otherwise the bank is able to recover only a fraction of the outstanding value, more precisely $R^\cC (M - X - I^{FC})^+$, where $R^\cC \in [0,1]$ is the recovery rate in case $\cC$ defaults. We are not concerned here with recovery rate modeling and so, for simplicity, we assume that $R^\cC$ is deterministic. In reality, recovery rates are an unknown random variable, not necessarily measurable with respect to $\sG_t$.
In summary, we see that the value at default, in case $\cC$ defaults, has the form
$$ \theta^\cC = X+I^{FC} + R^\cC (M - X - I^{FC})^+ + (M - X -I^{FC})^-.$$ 

Similarly, if the bank itself defaults, it has the right to proper fulfillment of the contract and hence, in addition to the collateral amount $X-I^{TC}$, it receives the outstanding balance $(M -X +I^{TC})^+$. If $M < X-I^{TC}$, the bank pays a fraction of its own obligation, namely $R^\cB (M - X +I^{TC})^-$, where $R^\cB \in [0,1]$ is the bank's (deterministic) recovery rate. This means the value at default, at the bank's own default, can be expressed as follows:
$$ \theta^\cB = X - I^{TC} + (M - X + I^{TC})^+ + R^\cB(M - X + I^{TC})^-.$$ 

As a consequence, the portfolio value at default (at time $\tau = \tau^\cB \wedge \tau^\cC$) is explicitly given by
\begin{equation} \label{theta}
\begin{split}
\theta_\tau 
&= 
\e1_{\tau^\cC < \tau^\cB} \theta^\cC_\tau + \e1_{\tau^\cB < \tau^\cC} \theta^\cB_\tau \\
&=\e1_{\tau^\cC < \tau^\cB} (X_\tau + I^{FC}_\tau + R^\cC \big(M_\tau-X_\tau - I^{FC}_\tau)^++(M_\tau-X_\tau - I^{FC}_\tau)^-\big) \\
& \quad + \e1_{\tau^\cB < \tau^\cC}\big(X_\tau - I^{TC}_\tau + (M_\tau-X_\tau + I^{TC}_\tau)^+ + R^\cB (M_\tau-X_\tau + I^{TC}_\tau)^-\big).
\end{split}
\end{equation}
To consider the credit risky portfolio value $\hat V$ prior to default, we set up a replicating portfolio including a cash account, which reflects funding and collateral exchange related cash flows.

\subsection{Dynamic portfolio replication with counterparty credit risk}

Classical pricing theory is developed around the assumption that market participants can freely borrow and lend, without the necessity of exchanging collateral, at a single interest rate, namely the riskless interest rate $r$. Here we take a more realistic approach and specify the different funding costs associated with different types of lending. Additionally, we include collateral margining and the banks interest earning/paying on different types of capital. Our goal is to build a self-financing replicating portfolio for the counterparty credit risky portfolio value $\hat V$, which we pursue in the next section. 

The replicating portfolio $\hat\Pi$ comprises of $\delta$ units of $S$, $\alpha^\cB$ units of $P^\cB$, $\alpha^\cC$ units of $P^\cC$, and $\varphi$ units of the vector of cash accounts $B$, where $\delta,\ \alpha^\cB,\ \alpha^\cC$ and $\varphi$ are stochastic processes. The vector of cash accounts $B$ is composed of several accounts each with their own rate of accumulation which will be discussed in more detail in the following. More precisely, the account $B$ decomposes into $n+6$ cash accounts which we write as vector
\begin{equation*}
B = (B^S,B^\cC,B^X,B^{TC},B^{FC},B^K,B^F)^\top. 
\end{equation*}
We assume that each cash account is default-risk-free and has the value 1 at $t=0$, e.g. $B^K_0=1$. The corresponding strategy $\varphi$ is given as
\begin{equation*}
\varphi = (\varphi^S,\varphi^\cC,\varphi^X,\varphi^{TC},\varphi^{FC},\varphi^K,\varphi^F)^\top.
\end{equation*}
We need the value of $\hat\Pi_t$ at each time $t \leq T\wedge\tau$ to replicate the value $\hat V_t$, i.e. $\hat\Pi_t = \hat V_t$, or equivalently
\begin{equation}
\begin{split}
\hat V_t &=\hat\Pi_t\\
& = \delta_t^\top S_t + \alpha^\cB_t P^\cB_t + \alpha^\cC_t P^\cC_t + \varphi_t^\top B_t
\end{split}
\end{equation}

\noindent 
\emph{Securities funding:} As explained in Section \ref{sec:valueV}, taking a position in the underlying $S$ requires entering a repo transaction. The transaction is fully collateralized and we must have 
\begin{equation} \label{securitiesFundingVhat}
\delta^\top_t S_t = - 	(\varphi^S_t)^\top B^S_t,
\end{equation}
where the dynamics of the cash account are again given by \eqref{dB^S}.

\noindent 
\emph{Counterparty bond funding:} Similarly, the bank enters into position in counterparty's bonds $P^\cC$ through a repo transaction, i.e. we must always have
\begin{equation} 
\alpha^\cC_t P^\cC_t = - \varphi^\cC_t B^\cC_t.
\end{equation}
The evolution of the repo cash account $B^\cC$ is given by 
$$dB^\cC_t = q^\cC_t B^\cC_t dt,$$
where $q^\cC$ is the repo rate for the bonds $P^\cC$. 

The bank and its counterparty have to satisfy regulatory and collateral requirements for bilateral transactions for which the rules are set forth by the CSA and government regulatory agencies.

\noindent 
\emph{Initial margin:} Initial margins are exchanged at the inception of the contract and held in segregated accounts which leaves them unaffected in case of a default event. Note that initial margins are not netted. The initial margin  $I^{FC} \geq 0$, that the counterparty posts with the bank, finances $\varphi^{FC}$ units of the margin account $B^{FC}$, i.e. we have
\begin{equation}
I^{FC}_t = \varphi^{FC}_t B^{FC}_t.
\end{equation}
The dynamics of $B^{FC}$ is given by 
\begin{equation*} 
dB^{FC}_t = - r^{FC}_t B^{FC}_t dt,
\end{equation*}
representing the interest rate $r^{FC}$ the bank pays it's counterparty on the initial margin received.

In the same way the initial margin to the counterparty, $I^{TC} \geq 0$, is held in $\varphi^{TC}$ units of another margin account $B^{TC}$, meaning we have
\begin{equation}
- I^{TC}_t = \varphi^{TC}_t B^{TC}_t.
\end{equation}
The dynamics of the account $B^{TC}$ is
\begin{equation*} 
dB^{TC}_t = - r^{TC}_t B^{TC}_t dt,
\end{equation*}
where $r^{TC}$ is the rate received by the bank for the initial margin posted. 

\noindent 
\emph{Variation margin:} Unlike the initial margin, the variation margin $X$ is usually fully rehypothecable which we will come back to when we consider the funding of the different margins. In case $X>0$, corresponding to the counterparty having posted collateral with the bank, it is then financed by $\varphi^X$ shares of the margin account $B^X$. The CSA rules dictate that an interest rate $r^X$ is to be paid to the counterparty. On the other hand, $X<0$ describes the bank's collateral posted to its counterparty into the account $B^X$, earning the rate $r^X$. To sum up, we always have 
\begin{equation}
X_t = \varphi^X_t B^X_t,
\end{equation}
and the dynamics of $B^X$ is 
$$dB^X_t = - r^X_t B^X_t dt.$$

\noindent 
\emph{Regulatory capital cash flows:} We include the cost of regulatory capital $K$ into the model. The capital is raised from equity and debt investors which amounts to holding $\varphi^K<0$ units of a cash account $B^K$, i.e. it holds
\begin{equation}
K_t = - \varphi^K_t B^K_t.
\end{equation}
Denoting the cost of capital with $r^K$, the dynamics of $B^K$ are
\begin{equation*}
dB^K_t = r^K_t B^K_t dt.
\end{equation*}

\noindent 
\emph{Funding of uncollateralized positions:} So far, we addressed the question of funding only for the positions in the underlying stock $S$, the default-risky counterparty bond $P^\cC$ and the regulatory capital $K$. Here, we deal with the funding of the gap between the derivative value $\hat V$ and the collateral. Recall that in the classical Black Scholes model, there is no collateral and the delta position in the underlying stock is financed using the risk free bank account. Hence the amount that needs to be funded is the difference between the derivative value and the delta position. 
In our case the collateral comprises of initial and variation margin, $I^{TC}$, $I^{FC}$ and $X$ respectively. The value $I^{TC}$ to be paid to the counterparty always needs to be funded, whereas the variation margin $X$ lowers the funding requirement if $X>0$ but otherwise raises it as well. Since the initial margin $I^{FC}$ from the counterparty is positive, it would lower the funding requirement too. However, initial margins are not rehypothecable and hence the value that needs to be funded is 
\begin{equation*}
\hat V_t -(X_t-I^{TC}_t).
\end{equation*}
The bank has two sources of funding. One way of financing is issuing its own bonds $P^\cB$, the other is external funding through a financing account. 

The value of the financing account $\varphi^F_t B^F_t$ together with the position $\alpha^\cB_t P^\cB_t$ hence always needs to be
\begin{equation} \label{fundequ}
\varphi^F_t B^F_t + \alpha^\cB_t P^\cB_t =  \hat V_t - X_t + I^{TC}_t.
\end{equation}
The dynamics of $P^F$ depend on whether the value $\varphi^F_t B^F_t$ is positive or negative. In the former case, this cash is invested at the riskless rate $r$ in order to not introduce further credit risk, whereas in the latter case funds are raised at the cost of $r^F$. Consequently, the evolution of $B^F$ is given by
 \begin{equation*}
 dB^F_t = r^\pm_t B^F_t dt,
 \end{equation*} 
where $r^{\pm}_t = r_t \e1_{ \{\varphi^F_t B^F_t >0\} } + r^F_t \e1_{ \{\varphi^F_t B^F_t <0\} }$. 

\bigskip 

Having specified the structure of the different cash accounts, we are now ready to define the dynamics of the portfolio $\Pi$ that replicates $\hat V$ prior to a default event. Since we require that the replicating portfolio $\hat \Pi$ is self-financing, i.e. any changes in the portfolio value arise exclusively from changes in the underlying instruments, we must have that 
\begin{equation}
\begin{split}
d \hat V_t &= d\hat\Pi_t\\
& = \delta_t^\top dS_t + \alpha^\cB_t dP^\cB_t + \alpha^\cC_t dP^\cC_t + \varphi_t^\top dB_t.
\end{split}
\end{equation}
Along with the dynamics for $S$, $P^\cB$ and $P^\cC$ and the above dynamics of the different cash accounts, we obtain
\begin{equation*}
	\begin{split}
	d \hat \Pi_t
	&=
	\delta_t^\top \left(\mu(t,S_t)dt + \sigma(t,S_t)dW_t \right) 
	+ \alpha^\cB_t \left(r^\cB_t P^\cB_t dt - P^\cB_{t-} dJ^\cB_t \right) 
	+ \alpha^\cC_t \left(r^\cC_t P^\cC_t dt - P^\cC_{t-} dJ^\cC_t \right) 
	\\
	& \quad
	- (\varphi_t^S)^\top \diag(\gamma^S_t-q^S_t) B^S_t  dt 
	+ \varphi_t^\cC q^\cC_t B^\cC_t dt
	- \varphi^X_t r^X_t B^X_t dt 
	\\
	& \quad 
	+ \varphi^{TC} r^{TC}_t B^{TC}_t dt 
	- \varphi_t^{FC} r^{FC}_t B^{FC}_t dt 
	+ \varphi_t^K r^K_t B^K_t dt 
	+ \varphi^F_t r_t^\pm B^F_t dt,
	\\
	&=
	\delta_t^\top \left(\mu(t,S_t)dt + \sigma(t,S_t)dW_t \right) \\
	& \quad
	+ \alpha^\cB_t \left(r^\cB_t P^\cB_t dt - P^\cB_{t-} dJ^\cB_t \right) \\
	& \quad
	+ \alpha^\cC_t \left(r^\cC_t P^\cC_t dt - P^\cC_{t-} dJ^\cC_t \right) \\
	& \quad
	+ \delta_t^\top \diag(\gamma^S_t-q^S_t) S_t  dt - q^\cC_t \alpha^\cC_t P^\cC_t dt \\
	& \quad
	+  \big( r^{TC}_t I^{TC}_t - r^{FC}_t I^{FC}_t - r^X_t X_t - r^K_t K_t \big) dt
	\\
	& \quad 
	+ \big( r_t (\hat V_t - X_t  + I^{TC}_t - \alpha^\cB_t P^\cB_t) + (r^F_t - r_t)(\hat V_t - X_t + I^{TC}_t- \alpha^\cB_t P^\cB_t)^{-}\big)dt,
	\end{split}
\end{equation*}
where in the last equality we have used formulas \eqref{securitiesFundingVhat}-\eqref{fundequ}. We can simplify the dynamics to
\begin{equation*}
	\begin{split}
		d \hat \Pi_t
		&=
		\delta_t^\top \sigma(t,S_t)dW_t
		- \alpha^\cB_t P^\cB_{t-} dJ^\cB_t
		- \alpha^\cC_t P^\cC_{t-} dJ^\cC_t  \\
		& \quad
		+ \Big( \delta_t^\top \big( \mu(t,S_t) +\diag(\gamma^S_t-q^S_t) S_t \big)
		+ \alpha^\cB_t (r^\cB_t-r_t) P^\cB_t
		+ \alpha^\cC_t (r^\cC_t-q^\cC_t) P^\cC_t
		\\
		& \quad
		+ (r^{TC}_t + r_t) I^{TC}_t 
		- r^{FC}_t I^{FC}_t 
		- ( r^X_t + r_t ) X_t 
		- r^K_t K_t
		+ r_t \hat V_t   
		\\
		& \quad 
		+ (r^F_t - r_t)(\hat V_t - X_t + I^{TC}_t- \alpha^\cB_t P^\cB_t)^{-}\Big)dt .
	\end{split}
\end{equation*}
This equation describes the dynamics of the replicating portfolio prior to $T\wedge\tau$.

\subsection{Fundamental BSDE}

Our next goal is to formulate the above replication problem in terms of a BSDE. To this end we set
\begin{equation} \label{ZUU}
\begin{split}
\hat Z_t^\top &= \delta_t^\top \sigma(t,S_t), \\
U^\cB_{t-} &= - \alpha^\cB_t P^\cB_{t-} ,\\
U^\cC_{t-} &= - \alpha^\cC_t P^\cC_{t-}.
\end{split}
\end{equation}
Since $\hat V=\hat \Pi$, we get
\begin{equation*}
\begin{split}
 d \hat V_t
&=
 \hat Z_t^\top dW_t
 +U^\cB_{t-} dJ^\cB_t
 + U^\cC_{t-} dJ^\cC_t
\\ & \quad
+ \Big( \hat Z_t^\top \sigma(t,S_t)^{-1} \big( \mu(t,S_t) +\diag(\gamma^S_t-q^S_t) S_t \big)
- (r^\cB_t-r_t) U^\cB_t
- (r^\cC_t-q^\cC_t) U^\cC_t
 \\
& \quad
+ (r^{TC}_t + r_t) I^{TC}_t 
- r^{FC}_t I^{FC}_t 
- ( r^X_t + r_t ) X_t 
- r^K_t K_t
+ r_t \hat V_t   
\\
& \quad 
+ (r^F_t - r_t)(\hat V_t - X_t + I^{TC}_t- \alpha^\cB_t P^\cB_t)^{-}\Big)dt .
\end{split}
\end{equation*}
Finally, defining the driver
\begin{equation} \label{fundamentalBSDEDriver}
\begin{split}
 g(t,s, & \hat v,  \hat z,u^\cB,u^\cC) = \\
&
 -\hat z^\top \sigma(t,s)^{-1} \big( \mu(t,s) +\diag(\gamma^S_t-q^S_t) s \big)
+ (r^\cB_t-r_t) u^\cB
+ (r^\cC_t-q^\cC_t) u^\cC  
\\
&
- (r^{TC}_t + r_t) I^{TC}_t 
+ r^{FC}_t I^{FC}_t 
+ ( r^X_t + r_t ) X_t 
+ r^K_t K_t
- r_t \hat v
\\
&   
- (r^F_t - r_t)(\hat v - X_t + I^{TC}_t +u^\cB)^{-}
\end{split}
\end{equation}
we can write the evolution of our hedging portfolio as the following BSDE,
\begin{equation} \label{jumpFundamentalBSDE}
\begin{split}
-d \hat V_t &=  g(t,S_t,\hat V_t,\hat Z_t,U_t^\cB,U^\cC_t) dt - \hat Z_t^\top dW_t  -U^\cB_{t-} dJ^\cB_t  - U^\cC_{t-} dJ^\cC_t , \quad t \in [0,\tau \wedge T],
\\
\hat V_{\tau \wedge T} &=  \e1_{ \tau > T } \mathfrak{N}(S_T) + \e1_{ \tau \leq T } \theta_{\tau}.
\end{split}
\end{equation}
We refer to this equation as the {\it fundamental BSDE} of counterparty credit risk modeling. Recall that the default value $\theta_\tau$ has been defined in Section \ref{sec:CloseOut}.

Unlike standard SDEs which in applications are supplemented by initial value conditions, a BSDE is posed with a terminal value condition. The terminal condition for the fundamental BSDE requires accounting for three possible outcomes. If neither the bank nor the counterparty default before the final maturity $T$, our process ends at $T$, with $\hat V_T = \mathfrak{N}(S_T) =V_T$. On the other hand, if a default occurs prior to $T$, the portfolio is closed out and the process terminates early. In this case, the final portfolio value $\theta_\tau$ depends on of which of the parties defaults first, and is given by \eqref{theta}. 

Note that the fundamental BSDE has a possible jump in the event of a default at time $\tau$, which makes its numerical implementation rather complex. Fortunately, there is an explicit mapping of this equation onto a continuous BSDE, which is conceptually clear and allows for a standard numerical implementation. This transformation is presented in detail in Appendix \ref{transformation}. Specifically, we show there that $\hat V_t$, $\hat Z_t$, $U^\cB_t$, and $U^\cC_t$ can be represented, for all $t \in [0, T]$, as
\begin{equation} \label{solFundamentalBSDE}
\begin{split} 
\hat V_t &= \hat \cV_t \e1_{t <\tau} +  \theta_\tau \e1_{t \geq \tau}, \\
\hat Z_t &= \hat \cZ_t \e1_{t \leq\tau} ,\\
U^\cB_t &= (\theta^\cB_t - \hat\cV_t) \e1_{t \leq \tau} ,\\
U^\cC_t &= (\theta^\cC_t - \hat\cV_t) \e1_{t \leq \tau} ,
\end{split}
\end{equation}
where the pair of processes $(\hat\cV_t,\hat\cZ_t)$ is the solution to the following BSDE:
\begin{equation} \label{redFundamentalBSDE}
\begin{split}
-d \hat\cV_t &=  g(t,S_t,\hat\cV_t,\hat\cZ_t,\theta^\cB_t - \hat\cV_t,\theta^\cC_t - \hat\cV_t) dt -\hat\cZ_t^\top dW_t, 
\\
\hat\cV_{ T} &= \mathfrak{N}(S_T), \quad t \in [0, T].
\end{split}
\end{equation}
We refer to this equation as the {\it reduced fundamental BSDE}. We emphasize here that it is enough to find the solution to the reduced BSDE \eqref{redFundamentalBSDE} and then use \eqref{solFundamentalBSDE} to find the credit-risky portfolio value $\hat V$. Hence \eqref{redFundamentalBSDE} will play a central role in the remainder of the paper.

As a simple yet instructive example of the above reduction, we consider the following jump BSDE:
\begin{equation} \label{easiestJumpBSDE}
\begin{split}
-dY_t &=  (\alpha Y_t +\beta U_t)dt - U_t dJ_t,\\
Y_{\tau \wedge T} &= \xi 1_{\tau>T} + \theta 1_{\tau \leq T},
\end{split}
\end{equation}
where $\alpha$, $\beta$, $\theta \in\bR$ are constants and $\xi$ is an $\mathcal F_{T}$-measurable random variable. 
We find a closed form solution to this BSDE by following the steps outlined above. The key is again to reduce the BSDE with random time horizon and a jump into a BSDE with fixed time horizon $T$ and without jumps. The corresponding reduced equation given by
\begin{equation}
\begin{split}
-d \cY_t &= (\alpha \cY_t + \beta (\theta -\cY_t)) dt,\\
\cY_T &= \xi,
\end{split}
\end{equation}
is a linear inhomogeneous ODE. Its solution reads
$$	\cY_t = \left(\xi + \frac{\beta \theta }{\alpha-\beta} \right) e^{(\alpha-\beta)(T-t)} - \frac{\beta \theta }{\alpha-\beta} \, .$$
As a consequence, applying Theorem \ref{thm:transformedBSDE} gives the explicit solution $(Y,U)$ of \eqref{easiestJumpBSDE} as
\begin{align*}
	Y_t &= \left( \left(\xi + \frac{\beta \theta }{\alpha-\beta} \right) e^{(\alpha-\beta)(T-t)} - \frac{\beta \theta }{\alpha-\beta}  \right) 1_{t\leq\tau} + \theta 1_{t > \tau},  \\
	U_t &= \left( \frac{ \alpha \theta }{\alpha-\beta} - \left(\xi + \frac{\beta \theta }{\alpha-\beta} \right) e^{(\alpha-\beta)(T-t)}\right) 1_{t \leq \tau}.
\end{align*}
Notice that it would be hard to solve \eqref{easiestJumpBSDE} directly without the reduction step.

\section{The Burgard-Kjaer PDE and the Feynman-Kac representation}
\label{sec:BKPDE}

In this section we derive the fundamental PDE in the spirit of Burgard and Kjaer, see \cite{bk12}, \cite{g15}. The PDE approach requires that all the rates and dividends introduced above are deterministic. In that sense the approach based on the fundamental BSDE \eqref{jumpFundamentalBSDE}, which requires only that the rates and dividends are adapted to $\sF$, is more general. 

We derive the Burgard-Kjaer PDE starting with the reduced fundamental BSDE \eqref{redFundamentalBSDE}. Namely, we make the following ansatz:
\begin{equation}
{\hat \cV}_t=u(t,S_t),
\end{equation}
where $u=u(t,s)$ is a smooth function $u:\,[0,T]\times\bR^n\to\bR$. Applying Ito's lemma we find that
\begin{equation} \label{itoLemma}
du (t,S_t)= \big(\d_t u(t,S_t) + \cL_t u(t,S_t) \big)dt + (\nabla_s u(t,S_t))^\top \sigma(t,S_t)dW_t.
\end{equation}
Here, the Markovian generator $\cL_t$ is defined by
\begin{equation}\label{markovGen}
\cL_t = \mu(t,s)^{\top} \nabla_s + \frac12 \operatorname{tr}(\sigma(t,s)\sigma(t,s)^{\top} \nabla_s^2),
\end{equation}
and $g$ is the driver defined in \eqref{fundamentalBSDEDriver}. 
As a consequence, we find that 
\begin{equation}
\hat \cZ_t = \sigma(t,S_t)^\top\,\nabla_s u(t,S_t),
\end{equation}
and hence $u$ satisfies the following terminal value problem:
\begin{equation} \label{fundamentalPDE}
\begin{split}
\d_t u + \cL_t u + g(t, s, u,\nabla_s u,\theta^\cB-u,\theta^\cC-u)&=0,\\
 u(T,s) &= \fN(s).
\end{split}
\end{equation}
The derivation above is standard, see e.g. \cite{pr14} for more details. Note that, explicitly equation \eqref{fundamentalPDE} takes the form
\begin{equation*}
\begin{split}
\partial_t u + \frac12\,\tr(\sigma\sigma^{\top} \nabla_s^2 u) - r u & = (\nabla_s u)^\top \diag(\gamma^S-q^S) s
- (r^\cB-r) (\theta^\cB - u)
\\
& \quad 
- (r^\cC-q^\cC) (\theta^\cC - u) 
+ (r^{TC} + r) I^{TC} 
- r^{FC} I^{FC}_t 
- ( r^X + r ) X 
\\
& \quad   
- r^K K 
+ (r^F - r)(\theta^\cB - X + I^{TC})^{-}  ,\\
u(T,s) &= \fN(s).
\end{split}
\end{equation*}
Recall that $\theta^\cB$ and $\theta^\cC$ depend explicitly on the portfolio close out value $M$, as explained in Section \ref{sec:CloseOut}. Specific cases for close out values $M$ are obtained along the lines of the arguments in Section \ref{sec:valueAdj} and coincide with the PDEs derived in \cite{bk12}. In particular, we see from the argument above that the process $\hat\cZ_t$ is essentially the delta of the portfolio. 

Once the solution to \eqref{fundamentalPDE} is established, the solution of the fundamental BSDE can explicitly be written as
\begin{equation} 
\begin{split} 
\hat V_t &= u(t,S_t) \e1_{t <\tau} +  \theta_\tau \e1_{t \geq \tau}, \\
\hat Z_t &= \sigma(t,S_t)^\top\,\nabla_s u(t,S_t) \e1_{t \leq\tau} ,\\
U^\cB_t &= (\theta^\cB_t - u(t,S_t)) \e1_{t \leq \tau} ,\\
U^\cC_t &= (\theta^\cC_t - u(t,S_t)) \e1_{t \leq \tau} .
\end{split}
\end{equation} 
From a practical perspective, this representation of the solution to the fundamental BSDE may be hard to use. Numerical algorithms for solving high dimensional PDEs tend to have poor performance characteristics. We believe that Monte Carlo simulations, discussed in Section \ref{numSec}, offer a more efficient and robust approach.  

Another consequence of \eqref{itoLemma} is the Feynman-Kac representation of the solution to \eqref{fundamentalPDE}. Namely, integrating \eqref{itoLemma} and taking the conditional expectation given the current state of the underlying $S$, we find that
\begin{equation}
\begin{split}
  u (t,s) =& \eE \Big[  \fN(S_T) \\
&  - \int_t^T  g(v, S_v, u(v,S_v),\nabla_s u(v,S_v),\theta_v^\cB-u(v,S_v),\theta_v^\cC-u(v,S_v)) dv \ \big| \ S_t=s\Big].
 \end{split}
 \end{equation}
Actually, the representation above is, in general, not an explicit representation of the solution to the PDE \eqref{fundamentalPDE}. Instead it is an alternative equation for \eqref{fundamentalPDE}, stated as an integral equation.

\section{Valuation adjustments}
\label{sec:valueAdj}

We now proceed to determining the valuation adjustment to the portfolio value accounting for the counterparty credit risk. Let $\cA_t$ denote the difference between the counterparty credit-risky portfolio value $\hat \cV$ and the risk neutral portfolio value $V$, i.e.
\begin{equation}\label{aDef}
\begin{split}
\cA_t &= \hat \cV_t - V_t.
\end{split}
\end{equation}
The total valuation adjustment $A_t$ is then given by $A_t = \cA_t \e1_{t < \tau}$. As before, we let $M$ denote the portfolio close out value at the time of default. Below we consider separately two commonly considered close out conventions, namely $M=V$ and $M=\hat V$. 

\subsection{Close out value $M=V$}

We first consider the case of the close out value $M$ being equal to $V$, namely the risk neutral portfolio value. This is the standard convention widely adopted in the industry, see \cite{bk12}, \cite{g15}. Fortuitously, the corresponding reduced fundamental BSDE turns out to be linear, and thus can be solved in closed form. 

To see this, we observe that the driver $g$ in \eqref{fundamentalBSDEDriver}, together with \eqref{theta}, is given by
\begin{equation} 
\begin{split}
g(t ,S,\hat \cV,& \, \hat \cZ,  \theta^\cB - \hat \cV, \theta^\cC - \hat \cV) =
\\
& 
-\hat \cZ^\top \sigma(t ,S)^{-1} \big( \mu(t ,S) + \diag(\gamma^S-q^S) S \big)
- (r^\cB + r^\cC - q^\cC) \hat \cV  
\\
&  
+ r^K K
- (r^{TC} + r^\cB) I^{TC}
+ (r^{FC} + r^\cC -q^\cC)  I^{FC} 
+ ( r^X + r^\cB +r^\cC - q^\cC ) X 
\\
&  
+(r^\cB - r) \big( (V-X+I^{TC})^+ + R^\cB(V-X+I^{TC})^- \big) \\
&
+(r^\cC - q^\cC) \big( R^\cC(V-X - I^{FC})^+ + (V-X-I^{FC})^- \big) \\
&
+(r^F - r) \big( V - X+I^{TC} \big)^- \ ,
\end{split}
\end{equation}
which is a linear function in $\hat \cV$ and $\hat \cZ$. Note that $V$ is given by \eqref{solV} and is simply an exogenous input to the equation. According to \ref{app:LinBSDE}, the corresponding reduced fundamental BSDE \eqref{redFundamentalBSDE} can be solved explicitly. To streamline the notation we first set
\begin{equation} \label{GDefinition}
\begin{split}
G_t	& = 
r^K_t K_t
- (r_t^{TC} + r_t^\cB) I^{TC}_t 
+ (r_t^{FC} + r^\cC_t -q^\cC_t)  I^{FC}_t 
+ ( r_t^X + r_t^\cB +r^\cC_t - q^\cC_t ) X_t 
\\
& \quad 
+(r_t^\cB - r_t) \big( (V_t-X_t+I^{TC}_t)^+ + R^\cB(V_t-X_t+I^{TC}_t)^- \big) \\
& \quad 
+(r^\cC_t - q^\cC_t) \big( R^\cC(V_t-X_t - I^{FC}_t)^+ + (V_t-X_t-I^{FC}_t)^- \big) \\
& \quad
+(r^F_t - r_t) ( V_t - X_t+I^{TC}_t )^-. 
\end{split}
\end{equation}
We also define the stochastic exponential 
\begin{equation*} 
\begin{split}
\hat \Gamma_{t,s} 
&= \mathcal E \Big( - \int_t^s  \big( \sigma(u,S_u)^{-1} \big(\mu(u,S_u) 
+ (\gamma^S_u-q^S_u) S_u \big) \big)^\top dW_u 
- \int_t^s (r^\cB_u + r_u^\cC - q^\cC_u) du \Big),
\end{split}
\end{equation*}
for $s \geq t$. Observe that as a consequence of \eqref{gammaDef} the following factorization property holds:
\begin{equation} 
\begin{split}
\hat \Gamma_{t,s} 
&=\Gamma_{t,s}\,\exp\Big(-\int_t^s  (r^\cB_u + r_u^\cC - q^\cC_u) du \Big).
\end{split}
\end{equation}
This shows that the impact of the counterparty risk on the time evolution of the fundamental BSDE consists in additional discounting. Using formula \eqref{solLinBSDE} in Appendix \ref{app:LinBSDE}, the solution to the reduced fundamental BSDE can thus be written as 
\begin{equation} \label{solRedBSDEMV}
\begin{split}
\hat \cV_t &=\eE_t\Big[ \hat \Gamma_{t,T} \fN(S_T) + \int_t^T \hat \Gamma_{t,s} G_s ds\Big] \\
&=\eE_t\Big[e^{-\int_t^T  (r^\cB_u + r_u^\cC - q^\cC_u) du} \,\Gamma_{t,T}\, \fN(S_T) + \int_t^T e^{-\int_t^s  (r^\cB_u + r_u^\cC - q^\cC_u) du} \,\Gamma_{t,s}\, G_s ds\Big].
\end{split}
\end{equation}
Note that the formula above is a natural extension of \eqref{solV}. 

Now, using the fact that $x^++x^-=x$, we find after some algebra that the valuation adjustment $\cA$ defined in \eqref{aDef} can explicitly be expressed as
\begin{equation}\label{totValAdjMV}
\begin{split}
\cA_t &= \eE_t\Big[  \big(D_{r^\cB + r^\cC - q^\cC}(t,T)  - D_{r}(t,T) \big) \Gamma_{t,T} \,  \fN(S_T)\Big]\\
&\quad - (1-R^\cC) \eE_t\Big[\int_t^T D_{r^\cB + r^\cC - q^\cC}(t,s)\, \Gamma_{t,s} \, (r_s^\cC - q_s^\cC)(V_s-X_s - I^{FC}_s)^+  ds\Big]\\
&\quad - (1-R^\cB) \eE_t\Big[\int_t^T D_{r^\cB + r^\cC - q^\cC}(t,s)\, \Gamma_{t,s} \, (r_s^\cB - r_s)(V_s-X_s+I^{TC}_s)^- ds\Big]\\
&\quad +\eE_t\Big[\int_t^T D_{r^\cB + r^\cC - q^\cC}(t,s)\, \Gamma_{t,s} \, r^K_s K_s ds\Big]\\
&\quad +\eE_t\Big[\int_t^T D_{r^\cB + r^\cC - q^\cC}(t,s)\, \Gamma_{t,s} \,\big( r_s^{FC} I^{FC}_s - (r_s^{TC} + r_s) I^{TC}_s + ( r_s^X + r_s ) X_s  \big)
 ds\Big]\\
&\quad +\eE_t\Big[\int_t^T D_{r^\cB + r^\cC - q^\cC}(t,s)\, \Gamma_{t,s} \,(r^F_s - r_s) ( V_s - X_s+I^{TC}_s )^- ds\Big],
\end{split}
\end{equation}
where $D_k(t,u)=e^{- \int_t^u k(v)dv}$ is the discount factor over the time interval $[t,u]$ using rate $k$. The first term on the right hand side of \eqref{totValAdjMV} reflects the difference in discounting in the classic Black-Scholes model and counterparty credit risky discounting as discussed above. The remaining terms on the right hand side can be identified as follows:
\begin{equation}
\mbox{CVA}_t= - (1-R^\cC) \eE_t\Big[\int_t^T D_{r^\cB + r^\cC - q^\cC}(t,s)\, \Gamma_{t,s} \, (r_s^\cC - q_s^\cC)(V_s-X_s - I^{FC}_s)^+ ds\Big],
\end{equation}
represents the credit valuation adjustment (CVA),
\begin{equation}
\mbox{DVA}_t= -(1-R^\cB) \eE_t\Big[\int_t^T D_{r^\cB + r^\cC - q^\cC}(t,s)\, \Gamma_{t,s} \, (r_s^\cB - r_s)(V_s-X_s+I^{TC}_s)^- ds\Big],
\end{equation}
represents the debt valuation adjustment (DVA),
\begin{equation}
\mbox{KVA}_t=\eE_t\Big[\int_t^T D_{r^\cB + r^\cC - q^\cC}(t,s)\, \Gamma_{t,s} \, r^K_s K_s ds\Big],
\end{equation}
represents the capital valuation adjustment (KVA),
\begin{equation}
\mbox{MVA}_t=\eE_t\Big[\int_t^T D_{r^\cB + r^\cC - q^\cC}(t,s)\,\Gamma_{t,s} \, \big( r_s^{FC} I^{FC}_s - (r_s^{TC} + r_s) I^{TC}_s + ( r_s^X + r_s ) X_s  \big)ds\Big],
\end{equation}
represents the margin valuation adjustment (MVA), and finally
\begin{equation}
\mbox{FVA}_t=\eE_t\Big[\int_t^T D_{r^\cB + r^\cC - q^\cC}(t,s)\, \Gamma_{t,s} \,(r^F_s - r_s) ( V_s - X_s+I^{TC}_s )^- ds\Big]
\end{equation}
is the funding valuation adjustment (FVA).

We now connect the solution of the reduced fundamental BSDE \eqref{solRedBSDEMV}, to the solution to the fundamental BSDE \eqref{jumpFundamentalBSDE}. Explicitly, we have the following relation between the two solutions:
\begin{equation} \label{solFundamentalLinearBSDE}
\begin{split} 
\hat V_t &= \hat \cV_t \e1_{t <\tau} +  \theta_\tau \e1_{t \geq \tau} ,\\
\hat Z_t &= \hat \cZ_t \e1_{t \leq\tau}, \\
U^\cB_t &= \Big(X_t + I^{FC}_t + R^\cC (V_t-X_t - I^{FC}_t)^+ + (V_t-X_t - I^{FC}_t)^- 
- \hat \cV_t \Big) \e1_{t \leq \tau}, \\
U^\cC_t &= \Big( X_t - I^{TC}_t + (V_t-X_t + I^{TC}_t)^+ + R^\cB (V_t-X_t + I^{TC}_t)^- 
- \hat \cV_t \Big) \e1_{t \leq \tau} .
\end{split}
\end{equation}
The expressions derived above extend the corresponding explicit formulas in \cite{bk12}.

\subsection{Close out value $M=\widehat{V}$}

Choosing the adjusted portfolio value $M=\widehat{V}$ as the close out value, we note that the generator of the reduced fundamental BSDE has the following form:
\begin{equation} \label{vHatBSDEDriver}
\begin{split}
g(t ,S,\hat \cV, & \, \hat \cZ,  \theta^\cB - \hat \cV, \theta^\cC - \hat \cV) =
\\
& 
-\hat \cZ^\top \sigma(t ,S)^{-1} \big( \mu(t ,S) + \diag(\gamma^S-q^S) S \big)  
+ r^K K
\\
& 
- (r^{TC} + r^\cB) I^{TC} 
+ (r^{FC} + r^\cC -q^\cC)  I^{FC} 
+ ( r^X + r^\cB +r^\cC - q^\cC ) X 
\\
&  
- (r^\cB + r^\cC - q^\cC) \hat \cV  
+(r^\cB - r) \big( (\hat \cV-X+I^{TC})^+ + R^\cB(\hat \cV-X+I^{TC})^- \big) \\
&
+(r^\cC - q^\cC) \big( R^\cC(\hat \cV-X - I^{FC})^+ + (\hat \cV-X-I^{FC})^- \big) \\
&
+(r^F - r) \big( \hat \cV - X+I^{TC} \big)^-.
\end{split}
\end{equation}
In contrast to the case of $M=V$, the resulting reduced fundamental BSDE is nonlinear in $\hat \cV$, and an explicit representation to its solution is not available.

Instead, we can construct an approximate solution, assuming that the counterparty credit risk adjustment $\cA$ is small relative to $V$. Specifically, we use the approximation:
\begin{equation}
\begin{split}
\hat{\cV}^+&=(V+\cA)^+\\
& \approx V^+ +\cA\,\e1_{V\geq 0}, \\
\hat{\cV}^-&=(V+\cA)^-\\
&\approx V^- +\cA\,\e1_{V < 0}.
\end{split}
\end{equation}
This approximation is first order accurate in $\cA$. Substituting $\hat{\cV}=V+\cA$ into \eqref{redFundamentalBSDE}, where $V$ satisfies the riskless equation \eqref{BSBSDE}, and using the above approximations, we obtain the following linear BSDE for the adjustment $\cA$,
\begin{equation}
\begin{split} 
- d\cA_t & = ( h^0_t - \fr_t \cA_t + \zeta_t^\top h_t^1 )dt -\zeta_t^\top dW_t,
\\ 
\cA_T &= 0.
\end{split} 
\end{equation}
Here, $h^0 = G + rV$, where $G$ is defined in \eqref{GDefinition}, $\fr$ is the following effective discounting rate:
\begin{equation*} 
\begin{split}
\fr &= r^\cB + r^\cC - q^\cC\\
& \quad -(r^\cB - r) ( (1-R^\cB) \e1_{V-X+I^{TC} \geq 0} + R^\cB) \\
& \quad -(r^\cC - q^\cC) ( (1 - R^\cC)  \e1_{V-X - I^{FC} < 0} + R^\cC ) \\
& \quad -(r^F - r)  \e1_{V-X+I^{TC} < 0}  \ , 
\end{split}
\end{equation*}
and 
$$h^1 = \sigma(t,S)^{-1} \big( \mu(t,S) + \diag(\gamma^S-q^S) S \big).$$

Using the results summarized in Appendix \ref{app:LinBSDE}, we can solve this linear BSDE explicitly. Namely, we define the following stochastic exponential:
\begin{equation*}
\begin{split}
\Gamma^A_{t,s} &= \mathcal E \Big( - \int_t^s \big(\sigma(u,S_u)^{-1} (\mu(u,S_u) 
+ \diag(\gamma^S_u-q^S_u) S_u ) \big)^\top dW_u 
- \int_t^s  \fr_u du \Big) \\
&= D_{\fr}(t,s) \Gamma_{t,s} .
\end{split}
\end{equation*}
Then formula \eqref{solLinBSDE} in Appendix \ref{app:LinBSDE} yields 
\begin{equation} 
\begin{split} 
\cA_t & = \eE_t\Big[\int_t^T\Gamma^A_{t,s} h^0_s ds\Big]\\
&= \eE_t\Big[\int_t^T D_{\fr}(t,s) \Gamma_{t,s} \big(G_s +r_sV_s\big) ds\Big].
\end{split}
\end{equation}
Explicitly, the expression above can be written as
\begin{equation}\label{totValAdjMVhat}
\begin{split}
\cA_t &= \eE_t\Big[\int_t^T D_{\fr}(t,s) \Gamma_{t,s}\,r_sV_s ds\Big]\\
&\quad - (1-R^\cC) \eE_t\Big[\int_t^T D_\fr(t,s)\, \Gamma_{t,s} \, (r_s^\cC - q_s^\cC)(V_s-X_s - I^{FC}_s)^+  ds\Big]\\
&\quad - (1-R^\cB) \eE_t\Big[\int_t^T D_\fr(t,s)\, \Gamma_{t,s} \, (r_s^\cB - r_s)(V_s-X_s+I^{TC}_s)^- ds\Big]\\
&\quad +\eE_t\Big[\int_t^T D_\fr(t,s)\, \Gamma_{t,s} \, r^K_s K_s ds\Big]\\
&\quad +\eE_t\Big[\int_t^T D_\fr(t,s)\, \Gamma_{t,s} \,\big( r_s^{FC} I^{FC}_s - (r_s^{TC} + r_s) I^{TC}_s + ( r_s^X + r_s ) X_s  \big)
ds\Big]\\
&\quad +\eE_t\Big[\int_t^T D_\fr(t,s)\, \Gamma_{t,s} \,(r^F_s - r_s) ( V_s - X_s+I^{TC}_s )^- ds\Big].
\end{split}
\end{equation}
The individual terms in this expression can be interpreted in a fashion similar to the analogous terms in \eqref{totValAdjMV}. Notice that, compare to \eqref{totValAdjMV}, the discount rate $r^\cB+r^\cC-q^\cC$ is replaced with $\fr$.

The first term on the right hand side of \eqref{totValAdjMVhat} is an artifact of the difference in discounting in the classic Black-Scholes model and the counterparty credit risky discounting. The remaining terms on the right hand side have the following interpretation:
\begin{equation}
\mbox{CVA}_t= - (1-R^\cC) \eE_t\Big[\int_t^T D_\fr(t,s)\, \Gamma_{t,s} \, (r_s^\cC - q_s^\cC)(V_s-X_s - I^{FC}_s)^+ ds\Big],
\end{equation}
represents the credit valuation adjustment (CVA),
\begin{equation}
\mbox{DVA}_t= -(1-R^\cB) \eE_t\Big[\int_t^T D_\fr(t,s)\, \Gamma_{t,s} \, (r_s^\cB - r_s)(V_s-X_s+I^{TC}_s)^- ds\Big],
\end{equation}
represents the debt valuation adjustment (DVA),
\begin{equation}
\mbox{KVA}_t=\eE_t\Big[\int_t^T D_\fr(t,s)\, \Gamma_{t,s} \, r^K_s K_s ds\Big],
\end{equation}
represents the capital valuation adjustment (KVA),
\begin{equation}
\mbox{MVA}_t=\eE_t\Big[\int_t^T D_\fr(t,s)\,\Gamma_{t,s} \, \big( r_s^{FC} I^{FC}_s - (r_s^{TC} + r_s) I^{TC}_s + ( r_s^X + r_s ) X_s  \big)ds\Big],
\end{equation}
represents the margin valuation adjustment (MVA), and finally
\begin{equation}
\mbox{FVA}_t=\eE_t\Big[\int_t^T D_\fr(t,s)\, \Gamma_{t,s} \,(r^F_s - r_s) ( V_s - X_s+I^{TC}_s )^- ds\Big]
\end{equation}
is the funding valuation adjustment (FVA).

Finally, we notice that the solution to the fundamental BSDE \eqref{jumpFundamentalBSDE} is related to the riskless portfolio value $V$ via the following approximation:
\begin{equation} \label{solVHatapproxBSDE}
\begin{split} 
\hat V_t & \approx  (V_t + \cA_t) \e1_{t <\tau} +  \theta_\tau \e1_{t \geq \tau} ,\\
\hat Z_t & = \hat \cZ \e1_{t \leq\tau}\\
& \approx (Z_t + \zeta_t) \e1_{t \leq\tau}, \\
U^\cB_t & \approx \Big(X_t + I^{FC}_t + R^\cC (V_t-X_t - I^{FC}_t)^+ + (V_t-X_t - I^{FC}_t)^- + V_t + \cA_t \Big) \e1_{t \leq \tau}, \\
U^\cC_t & \approx \Big( X_t - I^{TC}_t + (V_t-X_t + I^{TC}_t)^+ + R^\cB (V_t-X_t + I^{TC}_t)^- + V_t + \cA_t \Big) \e1_{t \leq \tau} ,
\end{split}
\end{equation}
where $\cA $ is the total adjustment calculated above. We emphasize that, unlike \eqref{solFundamentalLinearBSDE}, the relations above link the solution of the fundamental BSDE to the solution of reduced fundamental BSDE.

\section{Choice of pricing measure and risk factor reduction}
\label{sec:PricingMeasureFactorReduction}

Up until now we have not addressed the issue of choosing a pricing measure $\eP$. The bank's portfolio may consist of a large number of assets. In practice, each of the asset classes is valued under its own martingale measure, which in turn depends on the appropriate choice of numeraire. For example, swaptions are priced under the forward swap measure, while equity options are priced under the rolling bank account measure. From the pricing perspective this approach is fully consistent, the choice of numeraire does not affect model valuations. However, the choice of pricing measure is crucial from the enterprise risk management perspective. The risk of a portfolio composed of various assets is not the sum of the risks of its components, as the dependences between assets may reduce or increase the total risk. It is thus important that the Monte Carlo simulations are carried out under a common pricing measure. There is no natural way of aggregating the different martingale measures into one pricing measure for the entire portfolio. We will not discuss this issue in detail here. In our model specification we simply assume a pricing measure $\eP$ that is not a martingale measure but rather behaves like a historical (aggregate) measure. The choice of $\eP$ is determined by the bank's risk appetite, regulatory requirements, and other factors, see e.g. \cite{kgb15}, \cite{s13}.

Once the aggregate pricing measure $\eP$ has been selected, the next issue is model calibration. There are two categories of variables entering the model: (i) directly observable such as asset prices, interest rates, recovery rates, etc., and (ii) not directly observable variables, which have to be estimated from the market data, such as volatilities, correlations, default intensities, etc.. Generally, for the indirectly observable model inputs, parameters inferred from cross-sectional market prices are associated with various martingale measures, while parameters inferred from historical time series are associated with physical measures. Notice that parameters such as volatilities can be deduced from both types of calculations. Their numerical values will differ depending on whether they are calculated as market implieds or by means of maximum likelihood estimation. However, typically the only practical way of calculating correlations is from historic time series. For default intensities, the CDS market, whenever available, yields risk neutral default probabilities. For less liquid names, without a liquid CDS market, historical default data, such as Moody's DBS bank or various credit ranking models, can be used (see e.g. \cite{mdbs16}). 

Another practical issue is the choice of risk factors. A financial institution is likely to contain thousands of positions in a netting set, each of which subject to market and counterparty risk. From a practical perspective, an analysis of a system with such a large number of risk factors is infeasible. In order to bring the dimensionality of the problem to a manageable size, a methodology of reducing the number of risk factors is required.

In mathematical terms, we are facing the issue of approximating the solution to the following high-dimensional FBSDE:
\begin{equation}\label{genFbsde}
\begin{split}
dS_t&=\mu(t,S_t)dt+\sigma(t,S_t)dW_t,\\
S_0&=s_0,\\
-dY_t&=f(t,S_t,Y_t,Z_t)dt-Z^\top_t dW_t,\\
Y_T&=\xi(S_T).
\end{split}
\end{equation}
Under the usual Lipschitz conditions on the coefficients, standard results guarantee the existence and uniqueness of the solution $(S_t,Y_t, Z_t)$ to this system (see e.g. \cite{pr14}).

A common approach used in practice is principal component analysis (PCA)\footnote{Alternatively, one might use, as is common in equity markets, a factor analysis based risk model such as e.g. BARRA.}. Specifically, the instantaneous covariance of the price process has the spectral decomposition:
\begin{equation}\label{covSpecDec}
\sigma(t,S_t)^{\top}\sigma(t,S_t)=\sum_{1\leq i\leq n}\,\lambda_{i,t}P_{i,t},
\end{equation}
where $\lambda_{i,t}\geq 0$ are the eigenvalues ordered by size, and $P_{i,t}$ are the spectral projections. Generically, each of the eigenvalues is non-degenerate and each of the spectral projections defines a one-dimensional subspace. In general, the eigenvalues and spectral projections are stochastic and depend on the realization of the process $S_t$ and time $t$. The left hand side of \eqref{covSpecDec} is estimated from suitable market data, as discussed above.

Reduction of risk factors is practical if only a small number $F\ll n$ of eigenvalues explain the covariance matrix, i.e.
\begin{equation}\label{appCovSpecDec}
\sigma(t,S_t)^{\top}\sigma(t,S_t)\approx\sum_{1\leq i\leq F}\,\lambda_{i,t}P_{i,t},
\end{equation}
with
\begin{equation*}
\sum_{i> F}\,|\lambda_{i,t}|^2<\varepsilon^2,
\end{equation*}
where $\varepsilon$ is a given tolerance level. We thus consider the projection operator 
\begin{equation*}
P_t=\sum_{1\leq i\leq F}\,P_{i,t}
\end{equation*}
onto the subspace spanned by the eigenvectors corresponding to the first $F$ eigenvalues. The key assumption ensuring that practicality of the factor reduction methodology is that $P_t$ is stable, and so its range persists regardless of market conditions. We can formulate this requirement heuristically as
\begin{equation}
\begin{split}
P_t&\approx P\\
&:=\frac{1}{T}\,\int_0^T\eE[P_t]dt,
\end{split}
\end{equation}
i.e. $P_t$ is approximately equal to its average $P$ over time $T$. We refer to the orthonormal basis in $\bR^n$ defined by this projection as the principal factors. The existence of $P$ is a strong assumption and it is not true in a general mathematical set up. Rather, it is an empirical fact indicating that the financial markets are driven by a relatively small number of persistent economic factors.

We define the following quantity:
\begin{equation}\label{deltaDef}
\begin{split}
\Delta(t) &= \Big(\int_0^t \eE \big[ \tr \big(  (I_n-P) \sigma(u,S_u)^\top \sigma(u,S_u) \big) \big] \,du\Big)^{1/2}\\
&=\Big(\int_0^t \eE \big[ \tr \big(\sigma(u,S_u)^\top \sigma(u,S_u)-P\sigma(u,S_u)^\top \sigma(u,S_u)P \big) \big] \,du\Big)^{1/2}.
\end{split}
\end{equation}
In words, $\Delta(t)$ measures the average discrepancy between the true covariance of the assets and the truncated covariance given by the projection onto the principal factors.

The key fact, established below, is that the price process $S$ of $\fN$ can be, to a good degree of accuracy, explained in terms of the principal risk factors. Specifically, we consider the projection $PW_t$ of the Brownian motion $W_t$ on the principal factors. In general, $P W_t$ is not a Brownian motion. However, we can choose a standard $F$-dimensional Wiener process $\widetilde W_t$ such that
\begin{equation}
PW_t=U\widetilde W_t,
\end{equation}
where $U$ is a constant $n \times F$-matrix with the property that $P=UU^\top$ and $U^\top U=I_F$.

We now consider a system driven by the principal risk factors, more precisely  
\begin{equation}\label{assetAppDyn}
\begin{split}
d\widetilde S_t&=\mu(t,\widetilde S_t)dt+\sigma(t,\widetilde S_t)U d\widetilde W_t,\\
\widetilde S_0&=s_0,
\end{split}
\end{equation}
where the drift and diffusion coefficients are the same as in \eqref{assetDyn}. As this equation can be understood as an SDE with a new diffusion coefficient $\sigma(t,S_t)U$, existence and uniqueness of the solution are obvious. We expect that the solution to this SDE is approximately equal to the true process $S_t$.

We turn these intuitions into a mathematical statement as follows. For an adapted, matrix-valued process $X_t$ we introduce the following semi-norm:
\begin{equation}
\|X_t\|_2=\eE\big[\tr(X_t^\top X_t)\big]^{1/2},
\end{equation}
and the norm
\begin{equation}
\|X\|_{2,\infty}=\sup_{0\leq t\leq T}\|X_t\|.
\end{equation}
Then we have the following theorem.

\begin{theorem}\label{facRedThm}
Assume that $\mu$ and $\sigma$ are Lipschitz continuous:
\begin{equation}
\begin{split}
\|\mu(t,s)-\mu(t,\tilde s)\|_2&\leq L_\mu\|s-\tilde s\|_2,\\
\|\sigma(t,s)-\sigma(t,\tilde s)\|_2&\leq L_\sigma\|s-\tilde s\|_2,
\end{split}
\end{equation}
with constant $L_\mu$ and $L_\sigma$, and satisfy the standard growth conditions:
\begin{equation}
\|\mu(t,s)\|_2 + \|\sigma(t,s)\|_2 \leq G ( 1 + \|s\|_2),
\end{equation}
with $G$ constant. 

Then \eqref{assetAppDyn} has a unique strong solution, and
\begin{equation}
\|\widetilde S-S\|_{2,\infty} \leq
\sqrt 2 \Big(\int_0^T \Delta(u)^2 du\Big)^{1/2} \exp(\gamma T),
\end{equation}
where $\gamma$ is a constant.
\end{theorem}
The proof of this theorem is presented in Appendix \ref{app:facRedProof}.
The theorem above says that the price process $\widetilde S_t$ driven by the truncated risk factors indeed approximates the true price process $S_t$. The tightness of the approximation is given by $\sqrt{\int_0^T \Delta(u)^2 du}$, and it may degrade exponentially fast in the time horizon $T$.

We now consider the backward part of the system of equations \eqref{genFbsde} driven by the principal risk factors: 
\begin{equation} \label{TildeBSDE}
\begin{split}
- d \widetilde Y_t &= f(t,\widetilde S_t, \widetilde Y_t, \widetilde Z_t ) dt - \widetilde Z_t^\top U d\widetilde W_t,  \\
\widetilde Y_T &= \xi(	\widetilde S_T) .
\end{split}
\end{equation}
Note that, unlike \eqref{genFbsde}, the driving process in \eqref{TildeBSDE} is not a standard Brownian motion anymore but the martingale $U \widetilde{W_t}$. There are several theoretical and practical aspects to be considered when working with this equation, which we will address in \cite{lr16}. For instance, existence of a solution $(\widetilde Y, \widetilde Z)$ to the above equation can be shown, however the solution is not unique. To see this, recall that the process $\widetilde{Z}$ represents the delta hedging strategy, also compare \eqref{ZUU}. Reducing the risk factors leads to an incomplete market as one can not fully hedge one's position anymore. We may, for example, choose $\widetilde{Z}$ to be a minimum variance strategy. In order to prove uniqueness of the solution, we have to introduce another process, compare \cite{km97}.

In order to measure the discrepancy between the exact BSDE and its approximation, we find it convenient to introduce the following (semi-)norms for adapted, vector-valued processes:
\begin{equation*}
\begin{split} 
\|X_t\|_\beta &= \eE\big[e^{\beta t }\tr(X_t^\top X_t)\big]^{1/2},\\
\|X\|_{\beta,2} &= \Big( \int_0^T \|X_u\|^2_\beta \, du \Big)^{1/2},\\
\|X\|_{\beta,\infty} &= \sup_{0 \leq t \leq T} \|X_t\|_\beta .
\end{split}
\end{equation*}
The following theorem shows that the solution to equation \eqref{TildeBSDE} approximates the solution of the backward part of system \eqref{genFbsde}.

\begin{theorem}\label{facRedThmBackward}
Assume that the terminal value $\xi$ and the driver $f$ are Lipschitz continuous:
\begin{equation}
\begin{split}
|\xi(\tilde s)-\xi(s)| & \leq K_\xi \|\tilde s - s \|_2 \\
|f(t,\tilde s, \tilde y, \tilde z) - f(t, s,y,z)| &\leq K_f \big(\|\tilde s - s \|_2+\|\tilde y - y \|_2 + \|\tilde z - z \|_2 \big) . 
\end{split}
\end{equation}
Then there exist constants $c_1,c_2,c_3>0$, depending on the time horizon $T$, such that the following inequalities hold:
\begin{equation} \label{approx}
\begin{split}
\|\widetilde Y - Y \|_{\beta,\infty} & \leq c_1 \|\widetilde S - S \|_{2,\infty} ,\\
\|\widetilde Z_t - Z_t \|_{\beta,2} & \leq c_2 \|\widetilde S - S \|_{2,\infty} ,\\
\int_0^T e^{\beta u}\eE\big[Z^\top_u (I_n-P)  Z_u \big] du  & \leq c_3 \|\widetilde S - S \|_{2,\infty} .
\end{split}
\end{equation}
\end{theorem}
The proof of this theorem is presented in Appendix \ref{app:facRedProof}. 
Consequently, the solution $(Y,Z)$ can be approximated by the processes $(\widetilde Y, \widetilde Z)$ driven by the principal risk factors. Moreover, the third of the inequalities in \eqref{approx} shows that the residual portion of $Z$ is small.

\section{Numerical results}
\label{numSec}

In this section we discuss a general numerical framework for solving continuous FBSDEs using Monte Carlo methods. Such equations include the reduced fundamental BSDE discussed above. We propose an algorithm for finding the counterparty credit-risky value $\hat V$ as an application. The method is then illustrated in a simple example. 

\subsection{Discretizing FBSDEs}
\label{subsec:DiscretizingFBSDE}

We briefly review a method for discretization of the forward backward system
\begin{equation}\label{fbSystem}
\begin{split}
dS_t &= \mu(t,S_t) dt + \sigma(t,S_t) dW_t,\\
S_0 &=s_0 ,\\
-dY_t &= f(t,S_t,Y_t,Z_t)dt -Z_t^\top dW_t, \\
Y_T &=\xi(S_T) .
\end{split}
\end{equation}
This method is classic and has been originally proposed by Bouchaud and Touzi in \cite{bt04}. 

For the forward process $S$, we apply a standard discretization scheme (see e.g. \cite{kp92}) such as Euler's or Milstein's scheme (for the latter, assuming suitable integrability conditions). Let $\pi=\{ t_0=0<t_1< \ldots < t_m = T\}$ denote a regular time grid, where $\Delta_i=t_{i+1}-t_{i}$. In particular, for the Euler scheme, the approximation takes the form of the following discretized forward process
\begin{equation*}
\begin{split}
S^\pi_0 &= s_0 , \\
S_{i+1}:=S^\pi_{t_{i+1}} &= S^\pi_{t_{i}} + \mu(t_{i},S^\pi_{t_{i}}) \Delta_i + \sigma(t_{i},S^\pi_{t_{i}}) \Delta W_i,
\end{split}
\end{equation*}
where $i \in \{0, \ldots,m-1\}$ and $\Delta W_{{i}} = W_{t_{i+1}}-W_{t_{i}}$. 

In order to approximate the backward part of the FBSDE \eqref{fbSystem}, we set $S_i := S^\pi_{t_{i}}$, $Y_i := Y^\pi_{t_{i}}$ and $Z_i=Z_{t_i}^\pi$. This leads to the following system:
\begin{equation}
\begin{split}
\label{BSDEEuler}
Y_i &= Y_{i+1} + f(S_i,Y_{i}, Z_i) \Delta_i -  Z_i^\top \Delta W_{i}  .
\end{split}
\end{equation}
Starting with the terminal condition
$$ Y_m = \xi(S_m),$$
we proceed with finding $Y_i$ and $Z_i$ for all $i = m-1, \ldots, 0$. Note that the $Y_i$'s in \eqref{BSDEEuler} are not adapted and depend on $Z_i$. These two problems can be solved by taking conditional expectations which leads to 
\begin{equation*}
\begin{split}
Y_{i} &= \eE_i \left[ Y_{i}  \right] \\
& = \eE_i \left[ Y_{i+1} \right] +  f(S_i,Y_{i}, Z_i) \Delta_i,
\end{split}
\end{equation*}
where we have used the notation $\eE _{i}[ \, \cdot \, ]=\eE _{t_i}[ \, \cdot \, ]$.
This implicit scheme can transformed into an explicit scheme by 
\begin{equation*}
\begin{split}
Y_{i} & = \eE_i \left[ Y_{i+1} +  f(S_i,Y_{i+1}, Z_i) \Delta_i  \right].
\end{split}
\end{equation*}
In order to determine $Z_i$, we multiply \eqref{BSDEEuler} by an increment $\Delta W_{i}$ and take conditional expectations. This yields
\begin{equation*}
\begin{split}
0 = \eE_i \left[ Y_i \Delta W_{i}  \right]
&= \eE_i \left[ Y_{i+1} \Delta W_{i} \right] - Z_{i} \Delta_i,
\end{split}
\end{equation*}
and hence we obtain the following expression for $Z_{i}$:
\begin{equation*}
Z_{i}
= \frac{1}{ \Delta_i} \, \eE_i \left[ Y_{i+1} \Delta W_{i}   \right] .
\end{equation*}

We are thus led to the following discrete time scheme for solving the backward part of system \eqref{fbSystem}:
\begin{equation}\label{discrFbSystem}
\begin{split}
Y_m &= \xi(S_m),\\
Z_{i}&= \frac{1}{ \Delta_i}\,\eE_i \big[ Y_{i+1} \Delta W_{i} \big] ,\\
Y_{i}&= \eE_i \big[ Y_{i+1} +  f(S_i,Y_{i+1}, Z_i) \Delta_i  \big],\\
\end{split}
\end{equation}
for $i=m-1, \ldots, 0$. Note that simulating this system requires numerical estimation of the conditional expected values $\eE_i[\,\cdot \,]$. We discuss this issue in the following section. 

\subsection{Conditional expectations via a Longstaff-Schwartz regression}
\label{subsec:HermitePolynomials}

A practical and powerful method of computing the conditional expected values in \eqref{discrFbSystem} is the Longstaff-Schwartz regression method originally developed for pricing American options \cite{ls01} (see also \cite{b12}). We use a variant of this method that involves the Hermite polynomials. This choice is natural as expressions involving conditional expectations of Hermite polynomials of Gaussian random variables lead to convenient closed form expressions. 

Let $\mathrm{He}_k\ofx$, $k=0,1,\ldots$, denote the $k$-th normalized Hermite polynomial corresponding to the standard Gaussian measure $d\mu(x)=(2\pi)^{-1/2}\,e^{-x^2/2} dx$. For a multi-index $\underline{k}=(k_1,\ldots,k_n)$, where each $k_a$ is a nonnegative integer, we define
\begin{equation}
\mathrm{He}_{\underline{k}}(x)=\prod_{a=1}^n\,\mathrm{He}_{k_a}(x_a).
\end{equation}
These functions form an orthonormal basis for the Hilbert space $L^2 \big(\bR^n, \mu_n\big)$, where $\mu_n$ is the standard Gaussian measure in $n$ dimensions, $d\mu_n(x)=(2\pi)^{-n/2}\,e^{-x^\top x /2} d^n x$.

The key property of $\mathrm{He}_{\underline{k}}(x)$ is the following addition formula for $\chi \in [0,1]$ and $w,x \in \bR^n$:
\begin{equation} \label{HeAdd}
\mathrm{He}_{\underline{k}}(\sqrt{\chi} \, w+\sqrt{1-\chi}\, x)
=\sum_{\underline{0}\leq\underline{j}\leq\underline{k}}\, {\underline{k}\choose\underline{j}}^{1/2}\; \chi^{|\underline{j}|/2}(1-\chi)^{|\underline{k-j}|/2}\mathrm{He}_{\underline{j}}(w)
\mathrm{He}_{\underline{k-j}}(x).
\end{equation}
Consequently, integrating over $x$ with respect to $\mu_n$ yields the following conditioning rules:
\begin{equation}\label{condRule}
\begin{split}
\eE\big[\mathrm{He}_{\underline{k}}(\sqrt{\chi} \, w+\sqrt{1-\chi}\,x)\,|\,w\big] & =\chi^{|\underline{k}|/2}\mathrm{He}_{\underline{k}}(w), \\
\eE\big[\mathrm{He}_{\underline{k}}(\sqrt{\chi} \, w+\sqrt{1-\chi}\,x)\,x_a \,|\,w\big] & = \chi^{|\underline{k}-1|/2} (1-\chi)^{1/2} \, \frac{\d \mathrm{He}_{\underline{k}} (w)}{\d w_a}\,.\\
\end{split}
\end{equation}
Here, $w, x$ are independent $n$-dimensional standard normal random variables.
The latter rule is found using the addition formula \eqref{HeAdd} and orthonormality of Hermite polynomials with respect to the standard Gaussian measure.
We shall use these rules in order to estimate the conditional expected values in \eqref{discrFbSystem}.

We set $W_{t_i}=\sqrt{t_i}\,w_i$, for $i=1,\ldots, m$, where $w_i$ is an $n$-dimensional standard normal random variable. We notice that
\begin{equation}\label{orthDecomp}
w_{i+1}=\sqrt{\chi_i}\,w_i+\sqrt{1-\chi_i}\,X_i,
\end{equation}
where $\chi_i=t_i/t_{i+1}$, and where $X_i$ is standard normal and independent of $w_i$. In the following, we shall use this decomposition in conjunction with \eqref{condRule}.

Now, we assume the following linear architecture:
\begin{equation}\label{hermExp}
Y_{i+1}=\sum_{\underline{k}:\,|\underline{k}|\leq K}\,g_{\underline{k},i+1}\,\mathrm{He}_{\underline{k}}(w_{i+1}),
\end{equation}
where $K$ is the cutoff value of the order of the Hermite polynomials. This is simply a truncated expansion of the random variable $Y_{i+1}$ in terms of the orthonormal basis $\mathrm{He}_{\underline{k}}(w_{i+1})$. The values of the Fourier coefficients are estimated by means of ordinary least square regression. Then, as a consequence of the conditioning rule \eqref{condRule},
\begin{equation}\label{condExp}
\eE_i[Y_{i+1}]=\sum_{\underline{k}:\,|\underline{k}|\leq K}\,g_{\underline{k},i+1}\,\chi_i^{|\underline{k}|/2}\,\mathrm{He}_{\underline{k}}(w_i).
\end{equation}
In other words, conditioning $Y_{i+1}$ on $w_{i}$ is equivalent to multiplying its Fourier coefficients $g_{\underline{k}}$ by the factor $\chi_i^{|\underline{k}|/2}$.

In practice, the formula for $Z_i$ given by \eqref{discrFbSystem} is hard to use. Instead, we find an explicit expression using the Hermite architecture, which was performant in our experiments.
\begin{proposition}
The following identity holds:
\begin{equation}\label{mallAbl}
\begin{split}
Z_i &= \frac{\d}{\d W_i}\,\eE_i[Y_{i+1}]\\
&=\frac{1}{\sqrt{t_i}}\sum_{k\leq K} g_{k,i+1} \chi_{i}^{k/2}\, k \mathrm{He}_{{k-1}} (w_i).
\end{split}
\end{equation}
\end{proposition}

\begin{proof}
It is sufficient to establish \eqref{mallAbl} in the one-dimensional case. Using \eqref{orthDecomp} and \eqref{HeAdd} we readily find that
\begin{equation*}
\begin{split}
\eE_i[\mathrm{He}_k(w_{i+1}) \Delta W_i]&= \sqrt{\Delta_i}\eE_i[\mathrm{He}_k(\sqrt{\chi_i}\,w_i+\sqrt{1-\chi_i}\,X_i) X_i]\\
&=  \frac{\Delta_i}{\sqrt{t_i}} \,\chi_{i}^{k/2} \, \frac{\d \mathrm{He}_{{k}} (w_i)}{\d w_i}\,,\\
\end{split}	
\end{equation*} 
where we have also used the second of the identities \eqref{condRule}. Consequently, using \eqref{discrFbSystem}, we find that
\begin{equation*}
\begin{split}
Z_i &= \frac{1}{\Delta_i} E_i \big[ \sum_{k\leq K} g_{k,i+1} \mathrm{He}_k(w_{i+1}) \Delta W_i \big]\\
&=  \sum_{k\leq K} g_{k,i+1} \chi_{i}^{k/2}\, \frac{\d \mathrm{He}_{{k}} (w_i)}{\d W_i}. \\
\end{split}
\end{equation*}
Comparing this with \eqref{condExp}, we see that \eqref{mallAbl} holds.
\end{proof}

Now that we have found a practical representation for $Z_i$, we proceed calculating $Y_i$ in \eqref{discrFbSystem}. To this end, we repeat the calculations in \eqref{hermExp} and \eqref{condExp} with $Y_{i+1}$ replaced by $Y_{i+1} + f(S_i,Y_i,Z_i)\Delta_i$. 

\subsection{Numerical solution for $\hat V$}

In order to solve the fundamental BSDE numerically, we proceed as follows. First we select the number of risk factors as in Section \ref{sec:PricingMeasureFactorReduction}. We then generate $N$ paths of the multi-factor Brownian motion required to simulate the dynamics of the underlying portfolio. Using the spectral decomposition approach to generating the Brownian paths, a practical choice could be $N=10,000$. 

Next, we simulate the asset price process $S$ by solving the forward equation of \eqref{fbSystem}. We use the price process $S$ as an input to find the value $V$, given by \eqref{solV}, of the netting set subject to no counterparty credit risk. 

Another key input into the model are the default intensities $\lambda^\cB$ and $\lambda^\cC$. Choosing $\lambda^\cB$ and $\lambda^\cC$ deterministic is the simplest possible and commonly selected option. However, this does not allow one to model wrong/right way risk \cite{cd03}, \cite{greg15}, \cite{gg12}. On the other hand, modeling stochastic default rates requires a stochastic dynamic. A standard approach consists in modeling $\lambda^\cB$ and $\lambda^\cC$ as diffusion processes. The Brownian drivers of these diffusions are appropriately correlated with the Brownian motions driving the underlying asset $S$. The sign of the magnitude of these correlations allows one to quantify the impact of wrong way risk on the counterparty credit. Solving the diffusions for $\lambda^\cB$ and $\lambda^\cC$ and applying the acceptance rejection method then generates the default times $\tau^\cC$ and $\tau^\cB$.

Next, the reduced fundamental BSDE $\hat \cV$ \eqref{redFundamentalBSDE} is solved. Since the reduced fundamental BSDE is of the form \eqref{fbSystem}, the numerical methodology discussed in Sections \ref{subsec:DiscretizingFBSDE} and \ref{subsec:HermitePolynomials} can be applied directly. For practical purposes we can choose $K$, the maximum order of Hermite polynomials, to be a small integer $2 \leq K \leq 4$. This choice offers a reasonable balance between accuracy and performance of the computation.

Finally, we find the counterparty credit risky portfolio value $\hat V$ as a result of the preceding computations using formula \eqref{solFundamentalBSDE}.

\subsection{Numerical illustrations}

In this section we illustrate the numerical method discussed above by applying it to a simple BSDE with a known explicit solution. A more thorough analysis of the above method as applied to the fundamental BSDE will be presented in a separate publication, see \cite{llr16}. 

Specifically, consider the following nonlinear BSDE: 
\begin{equation} \label{BSDEnonlinear}
\begin{split}
-dY_t&= \left(\alpha Y_t + \beta|Z_t| + \gamma^\top U_t  - \gamma^\top \theta (\alpha - \gamma^\top \mathbb 1) (T-t) \right)dt - Z_t^\top dW_t - U_t^\top dJ_t ,\\
Y_{\tau \wedge T} &= e^{a^\top W_T} 1_{\tau>T} + (\theta^1 \e1_{\tau= \tau^1} + \ldots + \theta^n \e1_{\tau= \tau^n}) 1_{\tau \leq T},
\end{split}
\end{equation}
with a counting process $J_t = (\e1_{\tau^1 \leq t} , \ldots , \e1_{\tau^n \leq t})^\top$, the first default time $\tau = \tau^1 \wedge \ldots \wedge \tau^n$, and a constant real-valued vector $\theta = (\theta^1, \ldots, \theta^n)^\top$. Moreover $\alpha, \ \beta \in \mathbb R$ and $a, \gamma \in \mathbb R^n$. Note that this BSDE has a random time horizon $\tau$ at which a jump occurs. 

As discussed in Section \ref{transformation}, the BSDE can be reduced to one without the jump and with a fixed time horizon. According to Theorem \ref{thm:transformedBSDE} the reduced BSDE is given by
\begin{equation} \label{ReducedBSDEExample}
\begin{split}
-d \cY_t&= (\alpha \cY_t + \beta|\cZ_t| + \gamma^\top (\theta - \cY_t  ) - \gamma^\top \theta (\alpha - \gamma^\top \mathbb 1) (T-t) )dt- \cZ_t^\top dW_t ,\\
\cY_T &= e^{a^\top W_T} .
\end{split}
\end{equation}
This BSDE has an explicit solution, which reads
\begin{equation}\label{explictSol}
(\cY_t, \cZ_t) = \big(M_t , a (M_t - \gamma^\top \theta (T-t))\big),
\end{equation}
where
\begin{equation}
\begin{split}
M_t&= \exp \Big( \Big(\frac12 a^\top a + \beta |a| + \alpha -\gamma^\top \mathbb 1 \Big)(T-t) + a^\top W_t \Big) + \gamma^\top \theta (T-t),\\
\end{split}
\end{equation}
for all $t \in [0,T]$. The solution $(Y,Z,U)$ of \eqref{BSDEnonlinear} is now obtained from solution to the reduced BSDE as 
\begin{equation}
\begin{split}
Y_t&= M_t 1_{t < \tau } + (\theta^1 \e1_{\tau= \tau^1} + \ldots + \theta^n \e1_{\tau= \tau^n}) 1_{t \geq \tau},\\
Z_t & = a (M_t - \gamma^\top \theta (T-t)) 1_{t \leq \tau} ,\\
U_t &= (\theta - M_t) 1_{t \leq \tau}.
\end{split}
\end{equation}

We will now construct a numerical solution to the reduced BSDE \eqref{ReducedBSDEExample}. More precisely, we compare the numerical solution to its explicit solution \eqref{explictSol} in the case of $n=1$. We assume the time horizon of $T=1$, and choose the following values of the parameters:
\begin{equation*}
\begin{split}
a &= -1.2 ,  \\
\alpha &= 0.5, \\
\beta &= 0.1, \\
\gamma &= 2, \\
\theta &= 1. 
\end{split}
\end{equation*} 
We divide the time interval into $m=250$ subintervals and generate $N=20,000$ Monte Carlo paths. For estimating the conditional expected values we choose the Hermite architecture \eqref{hermExp}, \eqref{mallAbl} with $K=4$.

Figure \ref{traj} shows representative Monte Carlo trajectories simulating $Y$ and $Z$. Here, the black lines are the paths of the exact solution \eqref{explictSol}, while the red lines are the numerical approximations calculated according to the algorithm above. Notice that the approximate path of $Y$ is very close to the exact trajectory. However, the paths representing $Z$ differ more. Apparently the numerical solution to the $Z$ process of a BSDE converges slower than the numerical solution to the $Y$ process. 

\begin{center}
	FIGURE 1
\end{center}

\begin{figure}[h]
\centering
\includegraphics[scale=0.60]{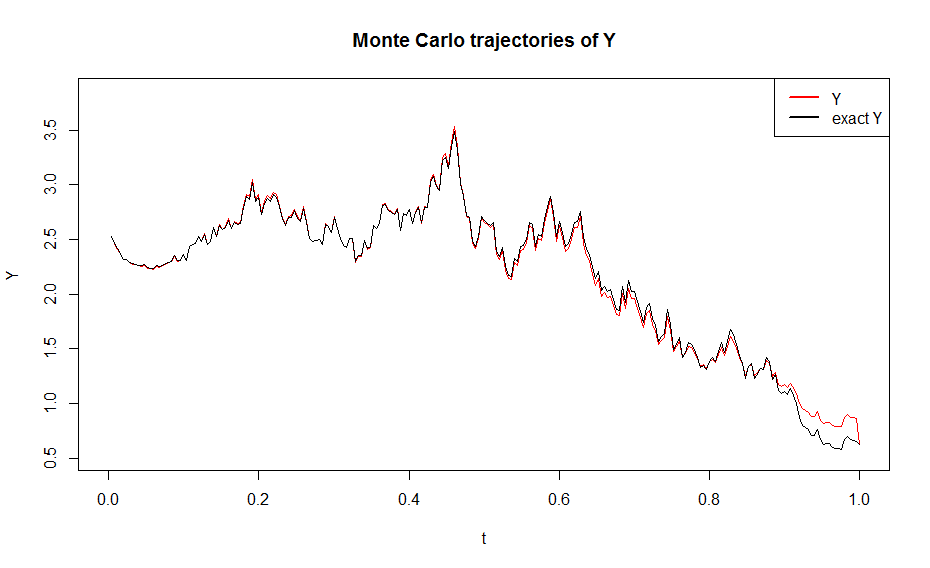}
\includegraphics[scale=0.60]{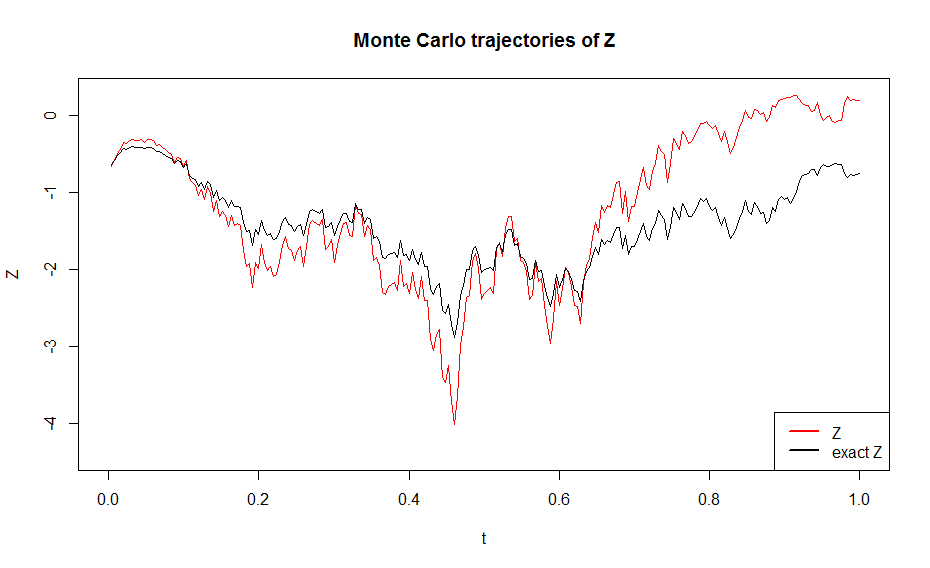}
\caption{Sample trajectories comparison of exact and numerical solution to \eqref{ReducedBSDEExample}. }
\label{traj}
\end{figure}

On the other hand, the expected values of both $Y$ and $Z$ are close approximations of the exact solution of the BSDE. This is shown in Figure \ref{ExpValues}. 

\begin{center}
	FIGURE 2
\end{center}

\begin{figure}[h]
\centering
\includegraphics[scale=0.65]{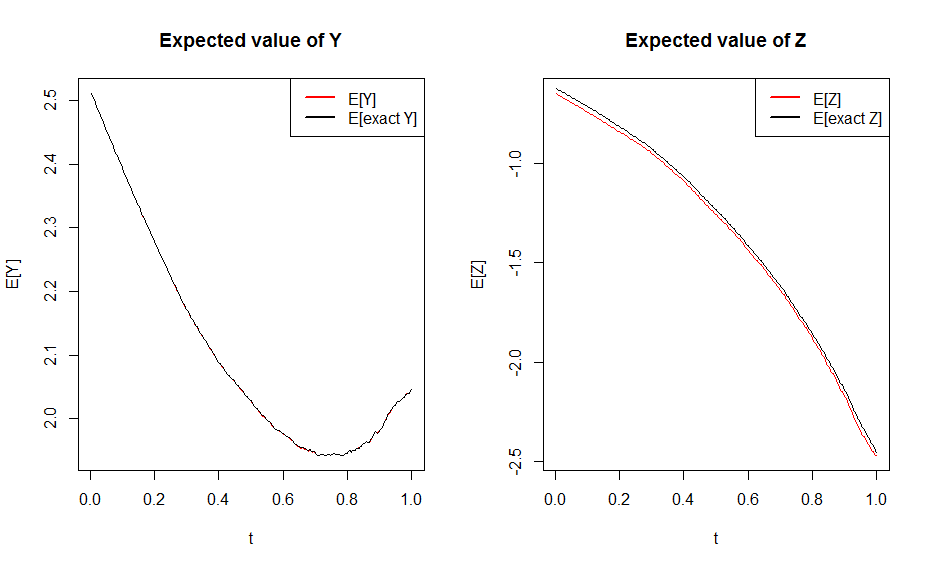}
\caption{Expected value comparison of exact and numerical solution to \eqref{ReducedBSDEExample}. }
\label{ExpValues}
\end{figure}

Finally, Figure \ref{RelError} shows the relative error of the expected values of $Y$ and $Z$ versus the expected values of the exact solution.

\begin{center}
	FIGURE 3
\end{center}

\begin{figure}[h]
\centering
\includegraphics[scale=0.47]{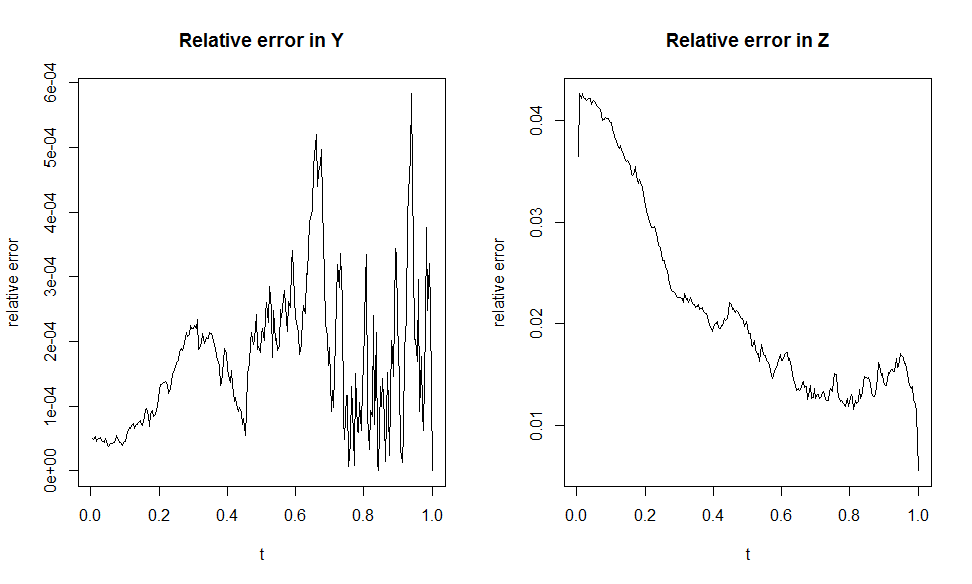}
\caption{Relative error of expected values of exact and numerical solution to \eqref{ReducedBSDEExample}.}
\label{RelError}
\end{figure}

\appendix
\section{Technical results and proofs}

In this section we present general results on SDEs and BSDEs we have used throughout the paper. First we are interested how the fundamental BSDE with jumps can be transformed into a reduced fundamental BSDE. 

\subsection{Transforming jump BSDEs into Brownian BSDEs} 
\label{transformation}

The fundamental BSDE can more generally be expressed as an equation
\begin{equation} 	\label{jumpBSDE}
\begin{split} 
- d Y_t &= f(t,Y_t,Z_t,U_t) dt - Z_t^\top dW_t -  U_t^\top dJ_t , \quad t \in [0, \tau \wedge T], \\
Y_{\tau \wedge T} &= \e1_{\tau >T} \xi + \e1_{\tau\leq T} (\theta^1_\tau \e1_{\tau= \tau^1} + \ldots + \theta^m_\tau \e1_{\tau= \tau^m}) ,
\end{split}
\end{equation}
driven by an $n$-dimensional Brownian motion $W$ and a counting process $J_t = (\e1_{\tau^1 \leq t} , \ldots , \e1_{\tau^m \leq t})^\top$. The solution to this BSDE is the set $(Y,Z,U)$ of adapted stochastic processes that satisfies \eqref{jumpBSDE}. The driver $f: \mathbb R_+ \times \mathbb R \times \mathbb R^n \times \mathbb R^m \to \mathbb R$ is a given deterministic function and $\tau = \tau^1 \wedge \ldots \wedge \tau^m$ denotes the first default time. The BSDE has a possible random time horizon, more precisely its terminal value depends on whether a default event happens before the fixed time horizon $T$. In that case the BSDE stops at the random time $\tau$ at an entry of the adapted stochastic process $\theta = (\theta^1, \ldots, \theta^m)^\top$. Otherwise the BSDE carries on to the final time $T$ with an $\mathcal F_T$- measurable random variable $\xi$ as the final value. Although jump BSDEs are rather complex to handle, we are in the particular situation that only ever one jump occurs and it happens at the very end of the BSDE. This is what we can use to transform the above random horizon jump BSDE into an equation without jumps and with a fixed terminal time, i.e.
\begin{equation} 	\label{brownianBSDE}
\begin{split} 
- d \cY_t &= f(t,\cY_t,\cZ_t,\theta_t-\cY_t) dt - \cZ^\top_t dW_t , \quad t \in [0,T], \\
\cY_T &= \xi.
\end{split}
\end{equation}
The following result is a generalization of Theorem 4.3 in \cite{kln13}, giving us the possibility to express the solution to the jump BSDE in terms of the solution of a continuous BSDE. 

\begin{theorem} \label{thm:transformedBSDE}
	If the pair of adapted stochastic processes $(\cY,\cZ)$ solves \eqref{brownianBSDE}, then the solution $(Y,Z,U)$ of \eqref{jumpBSDE} is given by
	\begin{equation} \label{sol}
	\begin{split} 
	Y_t &= \cY_t \e1_{t <\tau} +  (\theta^1_\tau \e1_{\tau= \tau^1} + \ldots + \theta^m_\tau \e1_{\tau= \tau^m})  \e1_{t \geq \tau} \ , \\
	Z_t &= \cZ_t \e1_{t \leq\tau} \ , \\
	U_t &= (\theta_t - \cY_t) \e1_{t \leq \tau}  .
	\end{split}
	\end{equation}
\end{theorem}

\begin{proof}
We consider three cases.

In the first case no default happens before the terminal time, i.e. $\tau > T$. On $\{\tau > T \}$ we have from \eqref{sol} that $Y_t = \cY_t$, $Z_t = \cZ_t$ and $U_t = \theta_t - \cY_t$ for all $t \in [0,T]$. As $(\cY,\cZ)$ solves \eqref{brownianBSDE}, we have 
\begin{equation*}
\begin{split} 
- d Y_t &= f(t, Y_t,Z_t,U_t) dt - Z^\top_t dW_t , \qquad t \in [0,T], \\
Y_T &= \xi = \e1_{\tau >T} \xi + \e1_{\tau\leq T} (\theta^1_\tau \e1_{\tau= \tau^1} + \ldots + \theta^m_\tau \e1_{\tau= \tau^m}).
\end{split}
\end{equation*}
on $\{\tau > T \}$. Additionally we know $\int_{t \wedge \tau}^{T \wedge \tau} U_s^\top dJ_s = 0$ on $\{\tau > T \}$ and hence we derive \eqref{jumpBSDE}.

In the second case a default happens between now and $T$, more precisely we look at $\{\tau \in  (t,T] \} =\{\tau >  t \} \cap \{\tau \leq T \}  $. Then again from \eqref{sol} we have on $\{\tau \in  (t,T] \}$ that $Y_s = \cY_s$, $Z_s = \cZ_s$, $U_s = \theta_s - \cY_s$ for all $s < \tau$. Using that $(\cY,\cZ)$ solves \eqref{brownianBSDE}, we obtain
\begin{equation*}
\begin{split} 
Y_t 
&= 
\cY_\tau + \int_t^\tau f(s, Y_s, Z_s, U_s) ds -  \int_t^\tau  Z^\top_t dW_t  \\
&=
(\theta^1_\tau \e1_{\tau= \tau^1} + \ldots + \theta^m_\tau \e1_{\tau= \tau^m})  
+ \int_t^\tau f(s, Y_s, Z_s, U_s) ds 
-  \int_t^\tau  Z^\top_t dW_t 
\\
& \qquad
- \big(  (\theta^1_\tau \e1_{\tau= \tau^1} + \ldots + \theta^m_\tau \e1_{\tau= \tau^m}) -\cY_\tau\big)
\end{split}
\end{equation*}		
for $t \in [0,\tau]$. The definition of $U$ from \eqref{sol} gives 
\begin{equation*}
\begin{split}
\int_t^\tau U_s^\top dJ_s &= U_\tau^\top (J_\tau-J_{\tau^-})\\
& =(\theta_\tau - \cY_\tau)^\top (J_\tau-J_{\tau^-})\\
&=  (\theta^1_\tau \e1_{\tau= \tau^1} + \ldots + \theta^m_\tau \e1_{\tau= \tau^m}) - \cY_\tau ,
\end{split}
\end{equation*}
meaning we have \eqref{jumpBSDE}.

The last case considers the situation when the default happens before or at time $t$, i.e. $\tau \leq t$. Again from \eqref{sol} we have $Y_t =  (\theta^1_\tau \e1_{\tau= \tau^1} + \ldots + \theta^m_\tau \e1_{\tau= \tau^n}) $ and thus on $\{ \tau \leq t\}$ we get
\begin{equation*}
\begin{split} 
Y_t 
&=
(\theta^1_\tau \e1_{\tau= \tau^1} + \ldots + \theta^m_\tau \e1_{\tau= \tau^m}) \\
&= 
\e1_{\tau >T} \xi + \e1_{\tau\leq T} (\theta^1_\tau \e1_{\tau= \tau^1} + \ldots + \theta^m_\tau \e1_{\tau= \tau^m})
+ \int_{t\wedge \tau }^{T\wedge \tau } f(s, Y_s, Z_s, U_s) ds \\
& \qquad
- \int_{t\wedge \tau }^{T\wedge \tau }  Z^\top_s dW_s  
- \int_{t\wedge \tau }^{T\wedge \tau } U_s^\top dJ_s ,
\end{split}
\end{equation*}
which is equation \eqref{jumpBSDE} in integral form.		 
\end{proof}

\subsection{Linear BSDEs} \label{app:LinBSDE}

Continuous linear BSDEs are equations with a driver that is linear in $Y$ and $Z$, meaning we consider equations of the type 
\begin{equation} \label{linBSDE}
\begin{split}
-dY_t&= (A_t + B_t Y_t + C_t^\top Z_t )dt - Z_t^\top dW_t, \quad t \in[0,T],\\
Y_T&=\xi , 
\end{split}
\end{equation}
with $n$-dimensional Brownian motion $W$, $\mathcal F_T$ measurable random terminal value $\xi$ and $A$, $B$, $C$ being adapted stochastic processes. The solution of this equation is any pair of adapted processes $(Y,Z)$ that satisfies \eqref{linBSDE}. These are some of the few BSDEs for which at least the first part of the solution $Y$ can be found explicitly. From \cite[Proposition 2.2]{kpq97} we have 
\begin{equation} \label{solLinBSDE}
\begin{split}
Y_t&= \eE_t \Big[ \xi \Gamma_{t,T} + \int_t^T A_s \Gamma_{t,s}  ds\Big] \\
\end{split}
\end{equation}
where 
\begin{equation}
\Gamma_{t,s}=\cE\Big(\int_t^s B_u du + C_u^\top dW_u\Big) \ .
\end{equation}
Here, $\cE (X)$ denotes the stochastic exponential of a stochastic process $X$. 

\subsection{Factor reduction}
\label{app:facRedProof}

In this section we prove Theorem \ref{facRedThm}.

\begin{proof} We begin by rewriting the SDEs for $S$ and $\widetilde S$ in the integral form:
\begin{equation*}
\begin{split}
S_t&=s_0+\int_0^t\mu(u,S_u)du+\int_0^t\sigma(u,S_u)dW_u,\\
\widetilde S_t&=s_0+\int_0^t\mu(u,\widetilde S_u)du+\int_0^t\sigma(u,\widetilde S_u)Ud\widetilde W_u.
\end{split}
\end{equation*}
Consequently, their difference is given by
\begin{equation*}
\begin{split}
\widetilde S_t-S_t=&\int_0^t\big(\mu(u,\widetilde S_u)-\mu(u, S_u)\big)du\\
&+\int_0^t\big(\sigma(u,\widetilde S_u)-\sigma(u, S_u)\big) U d \widetilde W_u+\int_0^t\sigma(u, S_u)\big(d  W_u-U d\widetilde W_u\big),
\end{split}
\end{equation*}
and thus, by means of Ito's isometry,
\begin{equation*}
\begin{split}
\|\widetilde S_t-S_t\|_2\leq&\|\int_0^t\big(\mu(u,\widetilde S_u)-\mu(u,S_u)\big)du\|_2\\
&+\|\int_0^t\big(\sigma(u,\widetilde S_u)-\sigma(u,S_u)\big) U d \widetilde W_u\|_2+\|\int_0^t\sigma(u,S_u)\big(d W_u- U d \widetilde W_u\big)\|_2\\
\leq&t^{1/2}\Big(\int_0^t\|\mu(u,\widetilde S_u)-\mu(u,S_u)\|_2^2\, du\Big)^{1/2}
+\Big(\int_0^t\| (\sigma(u,\widetilde S_u)-\sigma(u,S_u)) U\|_2^2\,du\Big)^{1/2}\\
&+ \Big(\int_0^t \eE \big[ \tr \big(  (I_n-P) \sigma(u,S_u)^\top \sigma(u,S_u) (I_n-P) \big) \big] \,du\Big)^{1/2}.
\end{split}
\end{equation*}
Note that 
\begin{equation*}
\begin{split}
\int_0^t \eE \big[ \tr \big(  (I_n-P) \sigma(u,S_u)^\top& \sigma(u,S_u) (I_n-P) \big) \big] \,du \\
& = \int_0^t \eE \big[ \tr \big(  \sigma(u,S_u)^\top \sigma(u,S_u) - P \sigma(u,S_u)^\top \sigma(u,S_u) P \big) \big] \,du \\
& = \int_0^t \Delta(u)^2 du,
\end{split}
\end{equation*}
where $\Delta(t)$ is defined by \eqref{deltaDef}. Using Lipschitz continuity, this yields
\begin{equation}\label{sstildeEst}
\|\widetilde S_t-S_t\|_2
\leq
C \Big(\int_0^t\|\widetilde S_u-S_u\|_2^2\, du\Big)^{1/2}
+\Big(\int_0^t \Delta(u)^2 du\Big)^{1/2},
\end{equation}
where $C$ is a constant, explicitly given as $C = L_\mu T^{1/2} + L_\sigma \|U\|$.

We shall now invoke classic Gr\"onwall's inequality: if $\varphi(t)$ is a nonnegative continuous function with
\begin{equation*}
\varphi(t)\leq\alpha(t)+\beta\int_0^t \varphi(s)ds,
\end{equation*}
where $\alpha(t)$ is a non-decreasing function and $\beta>0$, then
\begin{equation*}
\varphi(t)\leq\alpha(t)\exp(\beta t).
\end{equation*}
Squaring \eqref{sstildeEst}, and applying the inequality above to $\varphi(t)=\|\widetilde S_t-S_t\|_2^2$, we obtain
\begin{equation*}
\|\widetilde S_t-S_t\|_2
\leq
\sqrt{2} \Big(\int_0^t \Delta(u)^2 du\Big)^{1/2} \exp(\gamma t),
\end{equation*}
where we have set $\gamma=C^2$. Taking the supremum over $0\leq t\leq T$ yields the claim.
\end{proof}	 

We now turn to the proof of Theorem \ref{facRedThmBackward}.

\begin{proof}
The difference between \eqref{TildeBSDE} and the backward part of the system \eqref{genFbsde} is given by
\begin{equation*}
\begin{split}
\widetilde Y_t - Y_t 
&= \xi (	\widetilde S_T) - \xi (	S_T) + \int_t^T \big( f(u,\widetilde S_u, \widetilde Y_u, \widetilde Z_u ) - f(u,S_u, Y_u, Z_u ) \big) du \\
& \quad - \int_t^T ( \widetilde Z_u - Z_u)^\top U d\widetilde W_u - \int_t^T Z_u^\top ( U d\widetilde W_u - d W_u).
\end{split}
\end{equation*}	
In the following we adapt the arguments used to prove the existence of a solution to a BSDE (see \cite{kpq97}, \cite{pr14}). Applying Ito's formula to the process $e^{\beta t}(\widetilde Y_t - Y_t)^2$, where the constant $\beta > 0$ will be chosen later, yields
\begin{equation*}
\begin{split}
& e^{\beta t}(\widetilde Y_t - Y_t )^2\\
&= 
e^{\beta T}(\xi (\widetilde S_T) - \xi (S_T))^2 
+ 2 \int_t^T e^{\beta u}(\widetilde Y_u - Y_u)\big( f(u,\widetilde S_u, \widetilde Y_u, \widetilde Z_u ) - f(u,S_u, Y_u, Z_u ) \big) du \\
& \quad 
- 2 \int_t^T e^{\beta u}(\widetilde Y_u - Y_u) ( \widetilde Z_u^\top - Z_u^\top U)  d\widetilde W_u 
- 2 \int_t^T e^{\beta u}(\widetilde Y_u - Y_u) Z_u^\top (Ud\widetilde W_u - d W_u) \\
& \quad 
- \int_t^T e^{\beta u} (\widetilde Z_u - Z_u)^\top P (\widetilde Z_u - Z_u) du -  \int_t^T e^{\beta u}  Z^\top_u (I_n-P) Z_u du \\
& \quad 
- \beta \int_t^T e^{\beta u}(\widetilde Y_u - Y_u)^2 du.
\end{split}
\end{equation*}	
Taking expectations on both sides of this equation leads to the following identity:
\begin{equation*}
\begin{split}
& \|\widetilde Y_t - Y_t \|_\beta^2 + \beta \int_t^T \|\widetilde Y_u - Y_u\|_\beta^2 du + \int_t^T \|P (\widetilde Z_u - Z_u)\|_\beta^2 \,du  + \int_t^T e^{\beta u}\eE\big[Z^\top_u (I_n-P)  Z_u \big] du \\
& = 
\|\xi (\widetilde S_T) - \xi (S_T)\|_\beta^2 + 2 \int_t^T \eE\big[e^{\beta u} (\widetilde Y_u - Y_u)\big( f(u,\widetilde S_u, \widetilde Y_u, \widetilde Z_u ) - f(u,S_u, Y_u, Z_u ) \big)\big] du.\\
\end{split}
\end{equation*}
Using Lipschitz continuity of the terminal condition $\xi$ and driver $f$, we obtain that
\begin{equation*}
\begin{split}
& \|\widetilde Y_t - Y_t \|_\beta^2 + \beta \int_t^T \|\widetilde Y_u - Y_u\|_\beta^2 du + \int_t^T \|\widetilde Z_u - Z_u\|_\beta^2 \,du + \int_t^T e^{\beta u}\eE\big[ Z^\top_u (I_n-P) Z_u \big] du \\
& \leq K_\xi \|\widetilde S_T - S_T\|_\beta^2 + 2 K_f\int_t^T \eE\big[ e^{\beta u} |\widetilde Y_u - Y_u|\big( |\widetilde S_u - S_u| + |\widetilde Y_u - Y_u| + |\widetilde Z_u - Z_u| \big)\big] du.
\end{split}
\end{equation*}	
Using the elementary inequality 
$$2ab \leq a^2\lambda^2+\frac{b^2}{\lambda^2}\,,$$
where $\lambda >0$ is a constant, we find that
\begin{equation*}
\begin{split}
2|\widetilde Y_u - Y_u|\big( |\widetilde S_u - S_u| &+ |\widetilde Y_u - Y_u| + |\widetilde Z_u - Z_u| \big)\\
&\leq (3+\lambda^2)|\widetilde Y_u - Y_u|^2+|\widetilde S_u - S_u|^2+\frac{1}{\lambda^2}\,|\widetilde Z_u - Z_u|^2.
\end{split}
\end{equation*}
We thus arrive at the following key inequality:
\begin{equation*}
\begin{split}
 \|\widetilde Y_t - Y_t \|_\beta^2 &+ (\beta-K_f(3+\lambda^2) )\int_t^T \|\widetilde Y_u - Y_u\|_\beta^2\, du\\
&+ (1-K_f/\lambda^2)\int_t^T \|\widetilde Z_u - Z_u\|_\beta^2 \,du + \int_t^T e^{\beta u}\eE\big[ Z^\top_u (I_n-P) Z_u \big] du\\
& \leq K_\xi \|\widetilde S_T - S_T\|_\beta^2 + K_f\int_t^T \|\widetilde S_u - S_u\|^2_\beta  du.
\end{split}
\end{equation*}
Now, we choose $\lambda$ sufficiently large so that $1-K_f/\lambda^2>0$, and then subsequently we choose $\beta$ so that $\beta-K_f(3+\lambda^2)>0$.
\end{proof}

%

\end{document}